\ifdefined\ReaderVersion
\documentclass[onecolumn,11pt,tightenlines,notitlepage,amsmath,amssymb,aps,prx,
 floatfix,superscriptaddress]{revtex4-2}
\else
\documentclass[reprint,amsmath,amssymb,aps,prx,floatfix,
 superscriptaddress]{revtex4-2}
\fi

\usepackage{graphicx}
\usepackage{amsthm}
\usepackage{braket}
\usepackage{tikz}
\usetikzlibrary{positioning}
\usepackage{xcolor}
\usepackage{booktabs,capt-of}
\usepackage[hidelinks]{hyperref}
\usepackage{cleveref}
\crefname{figure}{Fig.}{Figs.}
\Crefname{figure}{Figure}{Figures}
\hypersetup{pdftitle={Exponentially Compressed and Garbage-Free Alias Sampling for Polynomial States},pdfauthor={Diyi Liu et al.}}

\providecommand{\texorpdfstring}[2]{#1}

\theoremstyle{plain}
\newtheorem{theorem}{Theorem}
\newtheorem{proposition}[theorem]{Proposition}
\newtheorem{lemma}[theorem]{Lemma}
\newtheorem{corollary}[theorem]{Corollary}
\theoremstyle{definition}
\newtheorem{definition}[theorem]{Definition}

\newtheorem{example}{Example}
\theoremstyle{remark}
\newtheorem{remark}[theorem]{Remark}

\newcommand{\PIA}{O_{\mathrm{PIA}}}
\newcommand{\Ot}{\widetilde O}
\newcommand{\Mcost}{\mathsf{M}}
\newcommand{\Had}{\mathsf{H}}
\newcommand{\keep}{\text{keep}}
\newcommand{\alias}{\text{alias}}
\usepackage{xcolor}

\usepackage[normalem]{ulem}

\newcommand{\figrev}[1]{#1}
\begin{document}

\title{Exponentially Compressed and Garbage-Free 
  Alias Sampling\\ for Polynomial State Preparation}

\author{Diyi Liu}
\email{diyiliu@lbl.gov}
\affiliation{Lawrence Berkeley National Laboratory, Berkeley, California
94720, USA}
\author{Hanyu Wang}
\affiliation{Department of Computer Science, University of California,
Los Angeles, California 90095, USA}

\author{Shuchen Zhu}
\affiliation{Microsoft Quantum, Redmond, Washington 98052, USA}
\affiliation{Department of Mathematics, Duke University, Durham,
North Carolina 27708, USA}
\affiliation{Duke Quantum Center, Duke University, Durham,
North Carolina 27701, USA}

\author{Jason Cong}
\affiliation{Department of Computer Science, University of California,
Los Angeles, California 90095, USA}
\author{Wibe Albert de Jong}
\affiliation{Lawrence Berkeley National Laboratory, Berkeley, California
94720, USA}
\author{David Williams-Young}
\affiliation{Microsoft Quantum, Redmond, Washington 98052, USA}
\author{Chao Yang}
\affiliation{Lawrence Berkeley National Laboratory, Berkeley, California
94720, USA}

\date{\today}

\begin{abstract}

Coherent alias sampling implement an important subroutine for block encoding and qubitization based quantum simulation using quantum lookup table.
Related constructions exploit data lookup and matrix structure to implement
block encodings
for use in quantum singular value transformation. For a state on $n$ qubits,
constructing and storing the alias table requires both classical and quantum resources exponential in $n$.
We show that when the amplitudes are samples of a polynomial of degree at
most $d$, the table can be represented and evaluated
coherently with cost polynomial in $n$, $d$, exponentially reduced from $\Theta(2^n)$ classical and quantum cost. We construct a quantum circuit for coherently computing the alias table for polynomial states using reversible integer arithmetic. By employing fast arithmetic, we construct a deterministic circuit for preparing polynomial states over Clifford+$T$ with a complete gate count of $\tilde{O}(ndB+B^2)$, where arithmetic width $B=\tilde{O}(dn+d\log(d+1))$. In the high precision limit for state fidelity $\epsilon$ with fixed $n$ and $d$, it is sufficient have $B=\Theta(\log(1/\epsilon))$ and the complete gate count becomes $O(\log^3(1/\epsilon)\log\log(1/\epsilon))$. Additionally, the effective representation of the alias table can also be used for classical sampling algorithms. Furthermore, we develop a novel technique for uncomputing the auxilary register generated by coherent alias sampling. This technique opens a research
direction in sampling based initial state preparation.

\end{abstract}

\maketitle

\section{Introduction}
\label{sec:introduction}

Recent advances in quantum hardware~\cite{krinner2022repeated,google2023surface,bluvstein2024logical,google2025qec,bluvstein2026faulttolerant}, together with progress in quantum error correction~\cite{shor1995qec,steane1996qec,dennis2002topological,bravyi2024memory,xu2024constant,zhou2025low,xu2025batched,cervia2025magic,cain2026shor}, have made fault-tolerant simulation of physical systems on quantum computers an increasingly realistic prospect.
Applications range from quantum many-body problems in condensed matter physics~\cite{liu2025lowgate,babbush2018,bassmanoftelie2022free,ding2024singleancilla,li2025dissipative,ding2025gibbs,hwang2026gibbs,ding2026thermalground}, high-energy physics~\cite{jordan2012qft,lamm2019general,ji2020gluon,ji2023gluon,gustafson2024primitive,gustafson2025synthesis,lamm2024block}, and nuclear physics~\cite{liu2025pairing,dumitrescu2018nucleus,lu2019subatomic,roggero2019response,roggero2020neutrino} to classical fluid dynamics~\cite{li2025fluid,gaitan2020fluid,liu2021dissipative,childs2021pde,meng2026turbulence}.
Many algorithms for these problems use block encodings~\cite{liu2025lowgate,liu2025pairing,camps2024explicit,chakraborty2019power,simon2025ladder,nibbi2024block,kane2025block,sunderhauf2024block,li2025binary,nguyen2022block,kharazi2025explicit,guseynov2024efficient,della2025efficient,camps2022fable,WanBlock,li2023efficient,camps2020approximate,chan2023simulating,zhang2026lowancilla,lemieux2024quantum,clader2022quantum,vasconcelos2025methods,dumitrescu2026matrix,schillo2026tare,yang2025dictionary,nibbi2026practical}, which embed a suitably normalized Hamiltonian or another matrix of interest as a block of a larger unitary.
A linear combination of unitaries (LCU) provides one way to construct such an encoding~\cite{berry2015taylor}.
Qubitization~\cite{low2019qubitization}, quantum signal processing (QSP), and quantum singular value transformation (QSVT)~\cite{gilyen2019qsvt,martyn2021grand,low2017qsp,dong2021phase,wang2022landscape,dong2024infinite,dong2024robust,ni2026fastphase,motlagh2024generalized,sunderhauf2023generalized,dong2026minimax,ni2025inverse,alexis2024nonlinear} provide tools for transforming encoded operators, including approximating time evolution.
A related formulation implements QSVT through Hamiltonian simulation using a single ancilla qubit~\cite{chakraborty2025qsvtwithoutblockencoding}.
For Hamiltonian simulation, algorithms based on these techniques achieve optimal or near-optimal query complexity in the corresponding oracle models~\cite{berry2015taylor,low2017qsp,low2019qubitization}.

A central subroutine in qubitization based method is quantum state preparation~\cite{groverrudolph,mottonen2005preparation,shende2006synthesis,plesch2011preparation,mozafari2022decision,gleinig2021sparse,babbush2018,lowkliuchnikov2024,gosset2026optimalt,liu2026resources,zhang2022optimaldepth,liluo2025sparsesize,mcardle2026quantum,sanders2019black,yuan2023optimal,gui2024spacetime,araujo2021divide,wang2024quantum,malz2024preparation,smith2024constant,ozols2013rejection,lemieux2024quantum,alonsolinaje2026compilation,berry2018improved,cottrell2019build,wu2019variational,zhang2024circuit,della2025efficient,fang2025highorder}, the task of implementing a unitary that maps the all zero state to a target state with classically specified amplitudes,
\begin{equation}
\label{equ:state_prep}
 \ket{0^n} \to \sum_{i=0}^{2^n-1} \alpha_i \ket{i},
 \qquad \sum_{i=0}^{2^n-1} |\alpha_i|^2 = 1.
\end{equation}
A recent compilation study found that alias sampling can reduce both the
$T$ count and total gate count when constructing a \textsc{prepare} oracle
on the instances studied~\cite{liu2026resources}. Classical alias sampling
starts with a uniformly chosen index and either keeps it or returns an
alternative called its \textit{alias}~\cite{walker1977alias,vose1991alias}.
A precomputed \textit{alias table} specifies the alternative and the
\textit{keep} threshold that determines the probability of retaining each
index. The alias table redistributes uniform probability to produce the desired
distribution to the prescribed accuracy. Once the alias table is constructed, generating a sample requires only a single uniform random draw, one table lookup, and one numerical comparison.
Babbush et al.~\cite{babbush2018} implemented alias sampling coherently
by replacing uniform random sampling with uniform superpositions and the aliasing with comparison and swap operations, after loading alias table with quantum lookup table techniques~\cite{lowkliuchnikov2024,zhu2025unified}. Although the method produces
a quantum state entangled with an additional register, rather than the clean
state in \cref{equ:state_prep}, it is compatible with block encodings for
qubitization-based methods~\cite{babbush2018}. However, the additional entangled register means that coherent alias sampling
does not directly prepare a clean initial state.

We are interested in polynomial
states, which can encode initial data for wave propagation and differential
equations. Let $N=2^n$. We consider states on $n$ qubits whose amplitudes are
proportional to the values of a real polynomial $p$ of degree at most $d$:
\begin{equation}
\label{eq:polystate}
 \sum_{i=0}^{N-1}\alpha_i\ket i,
 \qquad
 \alpha_i=\frac{p(i/N)}{\sqrt{\sum_{j=0}^{N-1}p(j/N)^2}}.
\end{equation}
The polynomial has only
$d+1$ coefficients, but the state has $2^n$ amplitudes. Without taking advantage of low dimensional description  and properties of finite degree polynomial, a quantum state preparation algorithm will result in exponential cost~\cite{gosset2026optimalt}, even proved optimal for general quantum states. This motivates the question:

\vspace{1.9mm}

 \textit{Can a alias sampling based method prepare a polynomial state on a $2^n$-point grid at $\operatorname{poly}(n)$ cost while leaving no auxiliary register entangled with the output?}
\vspace{1.9mm}

Efficient polynomial state preparation by alias sampling requires more
than returning the auxiliary registers to zero.
The usual alias table contains $N=2^n$ entries, even when the
polynomial has a short description. Constructing and storing the table,
as well as implementing a direct coherent lookup, requires resources
exponential in $n$. Quantum lookup techniques offer resource
tradeoffs~\cite{lowkliuchnikov2024,zhu2025unified}, but still rely on
an explicit list of entries.

Our first contribution is a compact description of an alias table whose
entries can be computed without constructing alias table of exponential size. A
distribution admits many valid alias tables.
We choose one that
redistributes probability in increasing index order. This ordering lets
us calculate a requested table entry from cumulative probability.
We obtain this description directly from
the polynomial coefficients, without first listing the $N$ entries.

Polynomial structure makes these cumulative probability efficient to compute. The indices
whose desired probabilities lie below, above, or at the uniform value
form only $O(d)$ groups of consecutive indices. Moreover, the sum of the
squared polynomial values over any interval can be evaluated from a
polynomial formula, without visiting its individual indices. These
properties let us replace explicit table lookup by reversible arithmetic.
Classical preprocessing, storage, and coherent evaluation of an entry all
have costs polynomial in $n$, $d$, and the coefficient bit length.

\begin{theorem}[Compact alias table, informal version of
\Cref{thm:implicit-representation}]
\label[theorem]{thm:intro-implicit-oracle}
Consider a polynomial state on $n$ qubits with degree bound $d$, where
$n,d\ge1$, and coefficients supplied in finite binary precision.
The ordered alias table has a classical description using $O(d(n+B))$
bits, computable in time polynomial in $n,d,B$. Here $B$ is a sufficient
bit width for the exact integer arithmetic. It can be chosen polynomial
in $n$, $d$, and the number of bits used to specify each coefficient.

Using this description, a quantum circuit can compute any table entry,
including when the input index is in superposition, using $\Ot(ndB)$
gates with fast arithmetic. Here $\Ot$ suppresses logarithmic factors.
No explicit list of the $2^n$ entries is constructed or stored.
\end{theorem}

\Cref{thm:intro-implicit-oracle} removes the need for an explicit table.
Our second contribution returns the auxiliary registers to zero. Because
the sampler starts with equal amplitudes, the auxiliary state for each
output is uniform over all sampling inputs producing that output. This
set is called its fiber. We assign consecutive integer ranks to these
inputs. The output and rank uniquely determine the original input, so
inverse ranking can reconstruct and clear the original input registers.
The remaining rank state is uniform over a consecutive range whose length
can be computed from the output. Inverse uniform state preparation then
clears the rank register without measurement or changing the output
amplitudes.

For a generic alias table, computing these ranks and their inverses may
require another large data structure. Our ordering avoids this problem:
each fiber has at most three blocks, whose boundaries are determined by
the cumulative totals already used to evaluate table entries. The same
compact description therefore supports efficient ranking and inverse
ranking. Sanders et al.~\cite{sanders2019black} used inverse uniform
preparation to remove a comparison record. Here, ranking puts the alias
sampler's auxiliary state into the uniform form needed for this erasure.
\Cref{thm:intro-clean-preparation} gives the complete preparation
cost, including these cleanup operations.

\begin{theorem}[Complete polynomial state preparation, informal version of
\Cref{thm:fast-clean-preparation,cor:precision-regimes}]
\label[theorem]{thm:intro-clean-preparation}
Using the compact description from \Cref{thm:intro-implicit-oracle},
we can prepare the target polynomial state $\ket\psi$ with a deterministic
Clifford+$T$ circuit, without postselection. For any $0<\epsilon<1$, its
joint output is within Euclidean distance $\epsilon$ of
\begin{equation}
 \ket\psi\otimes\ket{0}_{\mathrm{aux}},
 \label{eq:intro-clean-map}
\end{equation}
where $\ket{0}_{\mathrm{aux}}$ denotes all auxiliary registers in the
zero state. With fast arithmetic, the complete gate count is
\[
 \Ot\!\left(ndB+B\bigl(B+\log(1/\epsilon)\bigr)\log(B/\epsilon)\right).
\]
This includes uniform state preparation, coherent angle computation,
rotation synthesis, and uncomputation. Here $\Ot$ hides factors
polylogarithmic in $B$ and $\log(1/\epsilon)$. For inverse polynomial
accuracy in $B$, the bound simplifies to $\Ot(ndB+B^2)$.
\end{theorem}

With the polynomial and $n,d$ fixed, the gate count grows nearly cubically
in the number of accuracy bits. A total preparation error of at most
$\epsilon$ can be achieved using
$O(\log^3(1/\epsilon)\log\log(1/\epsilon))$ gates. This error includes both
rounding the real coefficients and approximating the circuit's rotations.

Our construction makes alias sampling an efficient route to clean
polynomial states. Related approaches to state preparation use conditional
rotations based on interval probabilities~\cite{groverrudolph},
matrix product states~\cite{gonzalezconde2024polynomial}, or quantum
singular value
transformation~\cite{gonzalezconde2024polynomial,mcardle2026quantum}.
Our exponential reduction is relative to an explicit alias table when
the degree and coefficient bit length are polynomially bounded in $n$.
The same compact description also supports classical alias sampling
without constructing the full table.

The remainder of the paper is organized as follows.
Section~\ref{sec:polynomial-input} introduces the integer input and the
sampling record. Section~\ref{sec:ordered-table} constructs the ordered
table, and Section~\ref{sec:implicit-oracle} gives its compact description
and coherent evaluation. Section~\ref{sec:clean-preparation} develops
fiber ranking and proves the complete state preparation bounds. The
appendices provide a worked example, circuit implementation details,
and bounds on input accuracy and arithmetic width.

\section{Polynomial states and the exact fragment model}
\label{sec:polynomial-input}


We begin with real polynomial amplitudes on a uniform grid. A common scaling
turns their dyadic approximations into integers, and the squares of these
integers become multiplicities. The normalized amplitudes themselves need not
be dyadic rational numbers. Throughout, $n\ge1$, $N=2^n$, and $d\ge1$; the
zero-qubit case is trivial.

\subsection{The integer polynomial model}

Let $[N]=\{0,\ldots,N-1\}$, where $N=2^n$. The exact input consists of
$n$, a degree bound $d$, and integers $a_0,\ldots,a_d$. Define
\begin{equation}
 P(i)=\sum_{k=0}^d a_kN^{d-k}i^k,\qquad i\in[N].
 \label{eq:integer-P}
\end{equation}
We assume that $P$ is nonzero on at least one grid point. Our target is
\begin{equation}
 \ket\psi=\frac{1}{\sqrt Z}\sum_{i\in[N]}P(i)\ket i,
 \qquad Z=\sum_{i\in[N]}P(i)^2>0.
 \label{eq:psi-P}
\end{equation}
The three quantities used in the counting construction are
\begin{equation}
 Q(i)=P(i)^2,\qquad Z=\sum_i Q(i),\qquad v_i=NQ(i).
 \label{eq:intro-weights}
\end{equation}
Here $Q(i)$ is an integer weight. The \textit{capacity} $Z$ is the number of
equally weighted fragment positions attached to each source index, and $v_i$
is the number of fragments that must be assigned to output $i$.
More generally, we write $[M]=\{0,\ldots,M-1\}$ for a nonnegative integer
$M$, with $[0]=\varnothing$. The coefficient bit length is
\begin{equation}
 \beta=\max_{0\le k\le d}\left\lceil\log_2(|a_k|+1)\right\rceil.
 \label{eq:coefficient-width}
\end{equation}

This integer input model can approximate the normalized amplitudes obtained
from polynomials with arbitrary real coefficients.
For $f(x)=\sum_{k=0}^d c_kx^k$ with nonzero sampled vector, set
$C=\max_k|c_k|>0$ and round $2^bc_k/C$ to an integer $a_k$. Then
\begin{equation}
 p(x)=\sum_{k=0}^d\frac{a_k}{2^b}x^k,
 \qquad P(i)=2^bN^d p(i/N).
 \label{eq:dyadic-p}
\end{equation}
Here $p$ approximates $f/C$, and the common positive factor $2^bN^d$ relating
$P(i)$ to $p(i/N)$ does not change the normalized state.
Appendix~\ref{app:input-accuracy} proves the resulting approximation error bound and
specifies how large $b$ must be to obtain $\epsilon$ precision. The construction below is exact for the
\emph{fixed integer polynomial} $P$; coefficient approximation and later
rotation synthesis contribute to separate errors.

\subsection{Modified capacity for alias sampling}

Alias sampling generates a discrete probability distribution from uniform
choices using a precomputed \textit{alias table}. For a positive integer
capacity $M$, the table assigns each source index $i\in[N]$ a threshold
$\keep(i)\in\{0,\ldots,M\}$ and an alternate index $\alias(i)\in[N]$.
The sampler chooses $i\in[N]$ and a position $z\in\{0,\ldots,M-1\}$
independently and uniformly, and returns $i$ if $z<\keep(i)$ and
$\alias(i)$ otherwise. The table is constructed so that these output
probabilities approximate the target distribution.

A common choice of capacity is
$M=2^\mu$ for some positive integer $\mu$. The resulting output probabilities
are integer multiples of $1/(N2^\mu)$, and their approximation error decreases
exponentially with $\mu$. The quantum implementation evaluates the sampling
rule coherently, retaining the source index $i$ and comparison position $z$
as garbage registers.
For a detailed introduction to coherent alias sampling and its logical
resource costs, see Refs.~\cite{babbush2018,liu2026resources}.

For the polynomial state considered here, we instead assign capacity $Z$ to each source index. We use the unitary
\begin{equation}
 U_Z\ket{0}=\frac{1}{\sqrt Z}\sum_{z=0}^{Z-1}\ket z.
 \label{eq:UZ-first}
\end{equation}
The desired output $y$ is assigned exactly $v_y=NQ(y)$ fragments, since
\begin{equation}
 \frac{v_y}{NZ}=\frac{Q(y)}{Z}=\frac{P(y)^2}{Z}.
 \label{equ:aliastory}
\end{equation}
Thus the replacement of the dyadic capacity by $Z$ is what makes the
combinatorial sampling distribution exact. It does not by itself remove
garbage, and it does not imply that $U_Z$ is exactly synthesizable over a
finite gate set. We keep these two issues separate below.

The normalization $Z=S(N)$, where $S(t)=\sum_{i=0}^{t-1}Q(i)$, is computed
once classically from the polynomial,
not summed over the $N$ samples during quantum execution. The implementation
of $U_Z$ uses the binary representation of this single integer; it is not a
lookup table containing $Z$ amplitudes. Section~\ref{subsec:uniform-implementation}
gives its circuit and gate count.

\subsection{Alias tables, fragments, and the source of garbage}

 \begin{definition}[Valid alias table]
An alias table consists of integers $\keep(i)\in\{0,\ldots,Z\}$ and
indices $\alias(i)\in[N]$, one pair for each $i\in[N]$.
It is valid if 
\begin{equation}
 \keep(y)+\sum_{\substack{i\in[N]\\\alias(i)=y}}\bigl[Z-\keep(i)\bigr]=v_y
\label{eq:balance-condition}
\end{equation}
for every $y\in [N]$.

 \end{definition}
 \begin{definition}[Fragment and shuffling map]
\label[definition]{def:fragment}
A \emph{fragment} is a pair $(i,z)$ with $i\in[N]$ and $z\in[Z]$.
An alias table induces the map
\begin{equation}
 \pi(i,z)=
 \begin{cases}
 i,         & z<\keep(i),\\
 \alias(i), & z\ge \keep(i).
 \end{cases}
\label{eq:fragment-map}
\end{equation}
\end{definition}
The \emph{fiber} $\pi^{-1}(y)$ is the set of all source fragments that produce
output $y$. Equivalently, validity means
\begin{equation}
	|\pi^{-1}(y)|=v_y,
	\qquad y\in[N].
\label{eq:fiber-size}
\end{equation}
We call $v_y$ the \emph{volume} of the fiber $\pi^{-1}(y)$.

There are $NZ$ equally weighted input fragments. Preparing their uniform
superposition and computing $\pi$ into a new output register gives
\begin{equation}
\begin{aligned}
 &\frac{1}{\sqrt{NZ}}\sum_{i,z}\ket{\pi(i,z)}\ket{i}\ket{z}\\
&=\sum_{y:v_y>0}\sqrt{\frac{v_y}{NZ}}
 |y\rangle|\mathrm{temp}_y\rangle\\
 &=\sum_{y:v_y>0}\frac{|P(y)|}{\sqrt Z}
 |y\rangle|\mathrm{temp}_y\rangle.
\end{aligned}
\label{eq:fragment-state}
\end{equation}
The state $|\mathrm{temp}_y\rangle$ is the uniform superposition of all source
pairs $(i,z)$ in the fiber $\pi^{-1}(y)$. It is a garbage state because it
retains the source information used to produce $y$ and generally depends on
$y$.
It cannot simply be discarded when a standalone pure polynomial state is
required. In fact, different fibers are disjoint, so the states
$\ket{\mathrm{temp}_y}$ are orthogonal. Tracing out the registers holding the
source index and fragment position gives
\begin{equation}
 \sum_y\frac{P(y)^2}{Z}\ket y\!\bra y,
 \label{eq:alias-reduced-state}
\end{equation}
not a coherent superposition with these square-root amplitudes.

The remedy is an invertible change of coordinates. Replace each source
fragment $(i,z)$ by its output $y=\pi(i,z)$ and an integer $t$ representing the ordering of the fragment within this fibre and referred to as \textit{rank}. We pause the discussion of details for the alias sampling algorithm for polynomial state but give an outline of the complete algorithm below.

We use $i\in[N]$ to denote the source index and $y=\pi(i,z)\in[N]$
to denote the output index. The four registers below hold, in order, the
output index $y$, source index $i$, source position $z$, and fiber rank $t$. Their widths are
$n,n,B,B$, where $B$ is chosen so that
$Z,v_y<2^{B}$ for every $y$. Only the ancilla qubits used for arithmetic
are omitted.

\begin{equation}
 \begin{gathered}
 \frac{1}{\sqrt{NZ}}\sum_{\substack{i\in[N]\\z\in[Z]}}
 \ket{0^n}\ket{i}\ket{z}\ket{0^{B}}\\
 \xrightarrow{\text{compute }\pi}
 \frac{1}{\sqrt{NZ}}\sum_{\substack{i\in[N]\\z\in[Z]}}
 \ket{\pi(i,z)}\ket{i}\ket{z}\ket{0^{B}}\\
 \xrightarrow[\text{uncompute } i,z]{\text{rank}}
 \frac{1}{\sqrt{NZ}}\sum_{\substack{y:v_y>0\\0\le t<v_y}}
 \ket{y}\ket{0^n}\ket{0^{B}}\ket{t}\\
 =\sum_{y:v_y>0}\frac{|P(y)|}{\sqrt Z}
 \ket y\ket{0^n}\ket{0^{B}}\,U_{v_y}\ket{0^{B}}\\
 \xrightarrow{\ \ket{v_y} \bra{v_y}\otimes U_{v_y}^{\dagger}}
 \sum_y\frac{|P(y)|}{\sqrt Z}
 \ket y\ket{0^n}\ket{0^{B}}\ket{0^{B}}\\
 \xrightarrow{\text{sign of }P(y)}
 \ket\psi\ket{0^n}\ket{0^{B}}\ket{0^{B}}.
 \end{gathered}
 \label{eq:technical-roadmap}
\end{equation}
Here $U_m\ket{0^{B}}=m^{-1/2}\sum_{t<m}\ket t$
for $0<m<2^{B}$.
The first arrow writes $y=\pi(i,z)$ into the initially zero $n$-qubit output register;
this same register holds $y$ throughout the remaining steps. The next arrow
computes the rank $t$ and uses unranking to reconstruct $(i,z)$ from $(y,t)$,
returning the source registers to $\ket{0^n}$ and $\ket{0^{B}}$;
both remain explicit in every subsequent line. These first two arrows separate the
logical steps for clarity; the implementation computes $y$ and $t$ together
with the forward ranking map, so no separate call to alias sampling is needed.
The circuit then computes $v_y$ into an ancilla register, uses it to
control the application of $U_{v_y}^{\dagger}$ to the position register,
and uncomputes $v_y$. This returns the position register to $\ket{0^B}$.

\figrev{Figure~\ref{fig:workflow} summarizes the preparation procedure,
separating the fixed classical data from the coherent operations and
making the two erasure steps explicit. First, the inverse ranking reconstructs the source pair from
$(y,t)$ and XORs it into the retained source registers, clearing $i$ and
$z$ while preserving $y$ and $t$. Second, the inverse uniform-state
preparation clears $t$ using the volume $v_y$ computed from $y$.
The classical data specify both steps without storing a separate record
for each of the $N$ indices. Figure~\ref{fig:opia} shows the coherent
evaluation of the implicit alias oracle called within the ranking steps.
The numerical illustrations in Figs.~\ref{fig:transport},
\ref{fig:weight-shuffle}, and \ref{fig:fiber-garbage} follow the same
eight-index example from weight redistribution to garbage removal.}

\begin{figure*}[!t]
\centering
\includegraphics[width=0.6\textwidth]{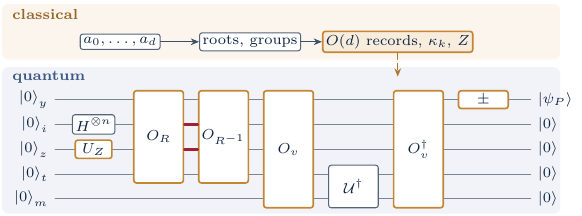}
\caption{Classical preprocessing and coherent preparation of a clean polynomial
state with the alias sampling method. The coefficients determine the classification groups, finite-difference
coefficients $\kappa_k$, and capacity $Z$; amber outlines identify operations
that use these fixed data. The quantum circuit prepares uniform source
registers with $H^{\otimes n}$ and $U_Z$. The ranking oracle $O_R$ writes
the output $y$ and fiber position $t$, and the inverse-ranking oracle
$O_{R^{-1}}$ reconstructs and XOR-clears the source registers $i,z$.
Red wire segments mark the retained sampling record between these operations.
The volume oracle $O_v$ computes $m=v_y=NP(y)^2$, the selected inverse
$\mathcal U^\dagger$ clears $t$, and $O_v^\dagger$ clears $m$.
The final sign phase, denoted $\pm$, converts $|P(y)|$ to $P(y)$.
Arithmetic scratch registers are omitted. The exact output shown assumes
ideal rotations; finite-precision guarantees are given in
Sec.~\ref{sec:clean-preparation}.}
\label{fig:workflow}
\end{figure*}

The next two sections construct the ordered table and evaluate it without
an $N$-entry lookup. Section~\ref{sec:clean-preparation} proves the steps for
ranking and removing garbage and accounts for the circuits preparing uniform states. Thus the central
questions can be read separately: what is the table, how is an entry
computed, and how is the garbage state in the registers holding the source fragment erased?
The main complexity result, \Cref{thm:fast-clean-preparation}, gives the
complete gate count using fast arithmetic, including angle computation and
rotation synthesis.

\par\medskip\noindent
\begin{minipage}{\columnwidth}
\captionof{table}{Notation used throughout. Ranks in the index lists and fiber positions use
different starting conventions.}
\label{tab:notation}
\small
\begin{tabular}{@{}p{0.15\columnwidth}p{0.78\columnwidth}@{}}
\toprule
Symbol & Meaning \\
\midrule
$N=2^n$ & Number of grid points. \\
$d,\beta$ & Polynomial degree bound and coefficient bit length. \\
$Z$ & Normalization $\sum_iP(i)^2$; also the capacity of each source index. \\
$v_y$ & Fiber volume $NP(y)^2$, which may exceed $Z$. \\
$r,s$ & Light/heavy list ranks, starting at $1$. \\
$t$ & Position in a fiber, starting at $0$. \\
$B$ & Sufficient magnitude width for exact arithmetic. \\
$\mu$ & Parameter specifying a dyadic capacity $M=2^\mu$. \\
\bottomrule
\end{tabular}
\end{minipage}\par\medskip

\section{The ordered alias table}
\label{sec:ordered-table}
We now construct a valid alias table by matching missing and extra weights in
increasing index order. Ordering is essential for the later ranking: the
light indices assigned to a given heavy index will form one interval of light
\emph{ranks}, even when their grid indices lie in several separated groups.
The construction in this section only uses nonnegative integer weights; the
polynomial structure is used in the next section to evaluate it efficiently.

\subsection{Aligning cumulative deficits and surpluses}
We first partition the output indices $y\in[N]$ according to whether their
target weights $v_y$ are smaller than, larger than, or equal to $Z$:
\begin{equation}
 \begin{aligned}
 \mathcal{I}_{\mathcal{L}}&:=\{y\in[N]:v_y<Z\},\\
 \mathcal{I}_{\mathcal{H}}&:=\{y\in[N]:v_y>Z\},\\
 \mathcal{I}_{\mathcal{B}}&:=\{y\in[N]:v_y=Z\}.
 \end{aligned}
\label{eq:LHB}
\end{equation}
We refer to these output indices as light, heavy, and balanced, respectively. Let
$m_{\mathcal{L}}:=|\mathcal{I}_{\mathcal{L}}|$ and
$m_{\mathcal{H}}:=|\mathcal{I}_{\mathcal{H}}|$. A
balanced output index $y$ requires exactly the $Z$ fragments supplied by the
corresponding source index $i=y$, so we set $\keep(y)=Z$ and
$\alias(y)=y$ for every $y\in\mathcal{I}_{\mathcal{B}}$.

Since $\sum_{y\in[N]}v_y=NZ$, the set $\mathcal{I}_{\mathcal{L}}$ is empty if
and only if $\mathcal{I}_{\mathcal{H}}$ is empty. In this case every output index is
balanced and the construction is complete.
Otherwise, $m_{\mathcal{L}},m_{\mathcal{H}}\ge 1$. We write the light and heavy output indices in increasing
order as
\begin{equation}
 \ell_1<\cdots<\ell_{m_{\mathcal{L}}},
 \qquad
 h_1<\cdots<h_{m_{\mathcal{H}}}.
\label{eq:ordered-lists}
\end{equation}
The rank of an index is its position in the corresponding ordered list. More
precisely,
\begin{equation}
 \operatorname{rank}_{\mathcal{I}_{\mathcal{L}}}(\ell_r):=r,
 \qquad
 \operatorname{rank}_{\mathcal{I}_{\mathcal{H}}}(h_s):=s.
\label{eq:index-ranks}
\end{equation}
These list ranks start from $1$. A light output needs fewer than the $Z$
fragments (each fragment has probability $1/(NZ)$ in the uniform input) initially assigned to every index, so that index can donate fragments.
A heavy output needs more than $Z$, so it must receive fragments.
The terms ``missing weight'' and ``extra weight'' below compare the target
weight with the capacity $Z$ of each source index; their transfer directions are opposite.
For each light index $\ell_r$, we define its missing weight $\delta_r$, and for
each heavy index $h_s$, we define its extra weight $e_s$:
\begin{equation}
 \delta_r:=Z-v_{\ell_r}>0,
 \qquad
 e_s:=v_{h_s}-Z>0.
\label{eq:missing-extra}
\end{equation}
The corresponding cumulative missing and extra weights are
\begin{equation}
 \begin{aligned}
 D(r)&:=\sum_{u=1}^{r}\delta_u,&D(0)&=0,\\
 E(s)&:=\sum_{u=1}^{s}e_u,&E(0)&=0.
 \end{aligned}
\label{eq:DE}
\end{equation}
Weight conservation gives
\begin{equation}
 D(m_{\mathcal{L}})=E(m_{\mathcal{H}}).
\label{eq:DE-conservation}
\end{equation}
Let $T:=D(m_{\mathcal{L}})=E(m_{\mathcal{H}})$. The intervals
$[D(r-1),D(r))$ of missing weight partition $[0,T)$, and the intervals
$[E(s-1),E(s))$ of extra weight also partition $[0,T)$. We call them blocks of
missing weight and extra weight, respectively.
\figrev{Figure~\ref{fig:transport} illustrates their alignment for
$N=8$, $P(i)=i$, and $Z=140$. The light indices $0,\ldots,4$ together
donate $460$ fragments, exactly the amount needed by the heavy indices
$5,6,7$. Assigning each entire light block according to its left endpoint
can oversupply a heavy output. The ordered construction compensates by
transferring fragments from that heavy index's own source column to the
next heavy output, as formalized below.}

\subsection{The table entries and their balance}

For each $1\le r\le m_{\mathcal{L}}$, let $s(r)$ denote the index of the
block of extra weight that contains the left endpoint of the $r$-th block of missing weight:
\begin{equation}
 s(r):=\min\{s\in\{1,\ldots,m_{\mathcal{H}}\}:E(s)>D(r-1)\}.
\label{eq:s-of-r}
\end{equation}
Similarly, for each $1\le s\le m_{\mathcal{H}}$, let $\rho(s)$ denote the index
of the first block of missing weight whose right endpoint reaches or passes the right endpoint
of the $s$-th block of extra weight:
\begin{equation}
 \rho(s):=\min\{r\in\{1,\ldots,m_{\mathcal{L}}\}:D(r)\ge E(s)\}.
\label{eq:rho-of-s}
\end{equation}
Both minima exist by \cref{eq:DE-conservation}. The strict inequality in
$s(r)$ and the non-strict inequality in $\rho(s)$ are intentional: they
specify which half-open block owns a shared endpoint. They must also be used
with these same conventions in the reversible searches.
We then define the ordered alias
table by
\begin{subequations}
\label{eq:ordered-table}
\begin{align}
 \keep(\ell_r)&=v_{\ell_r},
 &\alias(\ell_r)&=h_{s(r)},
 \qquad 1\le r\le m_{\mathcal{L}},
\label{eq:ordered-light}
\end{align}
\begin{equation}
 \keep(h_s)=Z+E(s)-D(\rho(s)),
 \qquad 1\le s\le m_{\mathcal{H}},
\label{eq:ordered-heavy-keep}
\end{equation}
\begin{equation}
 \alias(h_s)=
 \begin{cases}
 h_{s+1},&1\le s<m_{\mathcal{H}},\\
 h_s,&s=m_{\mathcal{H}},
 \end{cases}
\label{eq:ordered-heavy}
\end{equation}
\begin{align}
 \keep(i)&=Z,
 &\alias(i)&=i,
 \qquad i\in\mathcal{I}_{\mathcal{B}}.
\label{eq:ordered-balanced}
\end{align}
\end{subequations}

\begin{figure}[!tbp]
\centering
\includegraphics[width=\columnwidth]{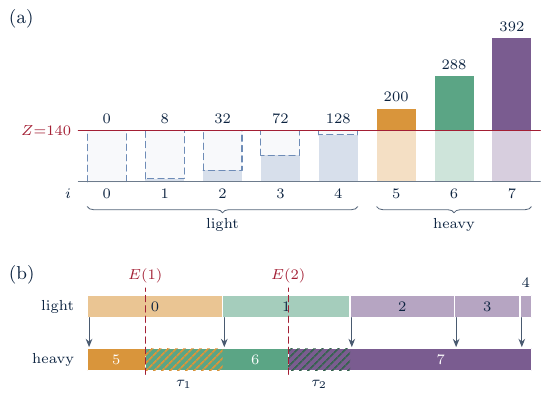}
\caption{Ordered weight redistribution for $N=8$ and $P(i)=i$.
(a) Target weights $v_i=8i^2$ relative to the common source capacity
$Z=140$. Dashed regions show the light indices' deficits below capacity;
the heavy indices require the excess above capacity.
(b) The missing-weight blocks of lengths $\delta_r$ and extra-weight
blocks of lengths $e_s$ are aligned in increasing index order; their
labels are grid indices. Each light block is assigned to the heavy block
containing its left endpoint. The boundaries $E(1)=60$ and $E(2)=208$
cut through light blocks, producing the hatched compensating transfers
$\tau_1=80$ and $\tau_2=64$ from heavy source indices $5$ and $6$.
The resulting table is shown in \cref{fig:weight-shuffle}.}
\label{fig:transport}
\end{figure}

The light index $\ell_r$ keeps $v_{\ell_r}$ fragments and assigns its other
$\delta_r$ fragments to $h_{s(r)}$. Define the number of fragments donated by the heavy index $h_s$ as
\begin{equation}
 \tau_s=Z-\keep(h_s)=D(\rho(s))-E(s).
 \label{eq:tau}
\end{equation}
The quantity $\tau_s$ measures how far the endpoint $D(\rho(s))$ of the
block of missing weight with index $\rho(s)$ extends past $E(s)$. The heavy source index $h_s$
transfers that many of its own fragments to $h_{s+1}$ when $s<m_{\mathcal H}$.
Since
\begin{equation}
 D(\rho(s)-1)<E(s)\le D(\rho(s)),
\end{equation}
we have $0\le\tau_s<\delta_{\rho(s)}\le Z$, and hence
$1\le\keep(h_s)\le Z$. The last heavy index has $\tau_{m_{\mathcal H}}=0$
by total weight conservation; its self-alias is never taken.

These are \emph{one-step} assignments. A fragment donated by a light index to
$h_s$ stays in the fiber of $h_s$. The compensation $\tau_s$ is donated by
source index $h_s$, not obtained by forwarding that arriving fragment.
The balance proof makes this distinction precise.

\begin{theorem}[Validity of the ordered table]
\label[theorem]{thm:ordered-valid}
The entries in \cref{eq:ordered-table} satisfy the balance condition in
\cref{eq:balance-condition} and therefore define a valid alias table for the
weights $v_i=NP(i)^2$.
\end{theorem}

\begin{proof}
Set $\rho(0):=0$ and $\tau_0:=0$, and fix
$a\in\{1,\ldots,m_{\mathcal{H}}\}$. From the definitions of $s(r)$ and
$\rho(a)$, we have
\begin{equation}
 \begin{aligned}
 s(r)=a
 &\quad\Longleftrightarrow\quad E(a-1)\le D(r-1)<E(a)\\
 &\quad\Longleftrightarrow\quad \rho(a-1)<r\le \rho(a).
 \end{aligned}
\label{eq:light-block-for-heavy}
\end{equation}
Thus the light indices assigned to $h_a$ are exactly
$\{\ell_r:\rho(a-1)<r\le \rho(a)\}$, which may be empty. The total weight
received by $h_a$ consists of $\tau_{a-1}$ fragments transferred from the
preceding heavy index when $a>1$, the missing weights transferred by these
light indices, and its own retained weight. For $a=1$, the first term is
$\tau_0=0$. Using \cref{eq:tau} for the first term, \cref{eq:DE} for the sum
of the missing weights, and \cref{eq:ordered-heavy-keep} for the retained weight,
we obtain
 \begin{align}
 &\tau_{a-1}
 +\sum_{\rho(a-1)<r\le \rho(a)}\delta_r
 +\keep(h_a)\notag\\
 &=\bigl[D(\rho(a-1))-E(a-1)\bigr]\notag\\
 &\quad+\bigl[D(\rho(a))-D(\rho(a-1))\bigr]\notag\\
 &\quad+\bigl[Z+E(a)-D(\rho(a))\bigr]\notag\\
 &=Z+e_a=v_{h_a}.
\label{eq:validity-cancellation}
\end{align}
No alias points to a light index, so $\ell_r$ receives
$\keep(\ell_r)=v_{\ell_r}$. Each balanced index receives all $Z$ of its
fragments.
Therefore, the balance condition holds for every index in $[N]$.
\end{proof}

\begin{example}[Computing the volume of a fiber]
For $N=8$ and $P(i)=i$, one has $Z=140$. The only source indices
contributing to output index $6$ are $1$, $5$, and $6$. The table gives
$\keep(1)=8$, $\alias(1)=6$, $\keep(5)=60$, $\alias(5)=6$, and
$\keep(6)=76$. The fiber of output index $6$ therefore contains
\begin{equation}
 (140-60)+(140-8)+76=288=8\cdot6^2
 \label{eq:main-heavy-example}
\end{equation}
fragments: $80$ from index $5$, $132$ from index $1$, and $76$ from index $6$.
\figrev{Figure~\ref{fig:weight-shuffle} displays these contributions as
the regions labeled $6$ in source columns $1$, $5$, and $6$.
Every column has height $Z$, and each fragment is assigned directly to
its final output by one keep-or-alias decision. No fragment is routed
through an intermediate output.
Appendix~\ref{app:worked-example} gives the complete table and its ranking.}
\end{example}

\begin{figure}[!t]
\centering
\includegraphics[width=\columnwidth]{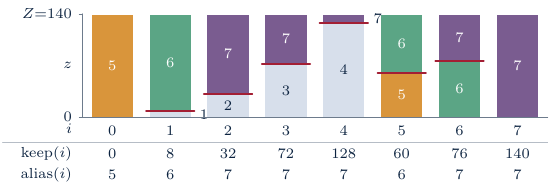}
\caption{The ordered alias table obtained from \cref{fig:transport}.
Each source column $i$ contains $Z=140$ equally likely positions $z$.
The red threshold separates the retained region $z<\keep(i)$ from
the region sent to $\alias(i)$; the two rows below list the table entries.
Labels inside the regions identify their output indices, with small
regions labeled alongside the columns. Amber, green, and purple identify
outputs $5$, $6$, and $7$; pale blue regions retain their own light index.
Summing the lengths of all regions with output label $y$ gives
$v_y=8y^2$, hence probability $v_y/(8\cdot140)=y^2/140$.
The table specifies one direct sampling decision, including for heavy
source columns.}
\label{fig:weight-shuffle}
\end{figure}

\section{An implicit alias oracle}
\label{sec:implicit-oracle}

The construction in Sec.~\ref{sec:ordered-table} reduces the computation of
each entry of the alias table to evaluating $D$ and $E$ and finding the ranks
$s(r)$ and $\rho(s)$. The resulting representation and coherent evaluation have costs
polynomial in $n,d,\beta$ and avoid constructing an explicit list of all
$N$ entries. We now
explain why these operations are efficient for polynomial weights. Instead of
storing $N$ individual entries,
we describe the light, heavy, and balanced indices using $O(d)$ groups of
consecutive indices and evaluate their cumulative weights in closed form.
There are two separate reasons for efficient compression: the stored description is
short while it is loaded on a quantum computer globally, and a quantum circuit can coherently evaluate the alias table from that description.  Here $S(t)$ is the cumulative sum of $Q$ over
$0,\ldots,t-1$, and $\PIA$ is the oracle that computes
$(\alias(i),\keep(i))$ from $i$. 

The computation diagram separates classical preprocessing from coherent
quantum computation:
\begin{equation}
 \begin{gathered}
 \text{Classical preprocessing}\\
 Q=P^2\ \longrightarrow\ \kappa_k
 \ \longrightarrow\ \text{formula for }S,\\
 Z=S(N),\qquad F=NP^2-Z,\\
 (F,S)\ \longrightarrow\ \text{group data}.\\[4pt]
 \text{Coherent quantum computation}\\
 \text{evaluate }S(t)\ \longrightarrow\ (D(r),E(s))\\
 \xrightarrow{\text{monotone searches}}\ (s(r),\rho(s))
 \ \longrightarrow\ \PIA.
 \end{gathered}
 \label{eq:oracle-dependency}
\end{equation}
Classical preprocessing computes the coefficients $\kappa_k$, the
normalization $Z$, and the group data, including cumulative weights at group
boundaries. These fixed data are compiled as constants into the quantum
circuit. During quantum execution, the circuit evaluates $S$, $D$, and $E$
as needed within the searches for $s(r)$ and $\rho(s)$, coherently on
arguments in superposition.

\subsection{Classification by a polynomial of low degree}

An index $i\in[N]$ is light, heavy, or balanced according to the sign of
\begin{equation}
 F(i):=v_i-Z=NP(i)^2-Z,
\label{eq:classification-polynomial}
\end{equation}
where $F$ is a polynomial of degree at most $2d$.

\begin{lemma}[Classification groups]
\label[lemma]{lem:classification-groups}
If $F$ is not the zero polynomial, $[N]$ can be divided into at most $4d+1$
groups of consecutive indices. All indices in each group have the same sign of
$F(i)$, while two neighboring groups have different signs. Moreover, the light
indices occur in at most $2d+1$ such groups, and the same bound holds for the
heavy indices. Finally,
$\mathcal{I}_{\mathcal{B}}$ contains at most $2d$ indices.
\end{lemma}

\begin{proof}
Let $r_F:=\deg F\le2d$. The polynomial $F$ has at most $r_F$ distinct real
roots. The intervals between these roots have constant nonzero sign and induce
at most $r_F+1$ groups of consecutive indices with nonzero sign on $[N]$. In
addition, the integer roots produce at most $r_F$ balanced groups. Thus $[N]$
contains at most $2r_F+1\le4d+1$ groups of constant sign. In particular, the light
indices occur in at most $r_F+1\le2d+1$ groups, and the same is true for the
heavy indices. The number of balanced indices is at most $r_F\le2d$.
\end{proof}

The additional groups in the $4d+1$ bound account for balanced integer
roots. We only use the fact that the number of groups grows linearly with
$d$. If $F$ is the zero
polynomial, every index is balanced and \cref{eq:ordered-balanced} gives the
complete alias table. We use the same trivial construction whenever all
indices in $[N]$ are balanced.

Let
\begin{equation}
 0=\gamma_0<\gamma_1<\cdots<\gamma_G=N,
 \qquad G\le4d+1,
\label{eq:group-boundaries}
\end{equation}
be the endpoints of the classification groups, so the $j$-th group consists of
the indices $i$ satisfying $\gamma_{j-1}\le i<\gamma_j$. For each group, we
record its type
$\theta_j\in\{\mathcal{L},\mathcal{H},\mathcal{B}\}$. We also record how many
light and heavy indices occur before the group:
\begin{equation}
 \begin{aligned}
 \lambda_j&:=|\{i\in\mathcal{I}_{\mathcal{L}}:i<\gamma_{j-1}\}|,\\
 \eta_j&:=|\{i\in\mathcal{I}_{\mathcal{H}}:i<\gamma_{j-1}\}|.
 \end{aligned}
\label{eq:group-earlier-counts}
\end{equation}
Thus $\lambda_j$ is the number of light indices smaller than
$\gamma_{j-1}$, and $\eta_j$ is the corresponding number of heavy indices.
We also record $D(\lambda_j)$, $E(\eta_j)$, and the total numbers of light and heavy indices
$m_{\mathcal{L}}$ and $m_{\mathcal{H}}$,
which are the upper bounds in the searches for $s(r)$ and $\rho(s)$.

To compute the group boundaries, we use \emph{root isolation}: finding
disjoint intervals with rational endpoints that contain exactly one distinct
real root each and together contain every real root. Each such interval is
called an \emph{isolating interval}. For this step only, we apply a certified
classical algorithm~\cite{sagraloff2011} to the
\emph{square-free part} of $F$, obtained by removing repeated copies of its
factors. This polynomial has the same distinct roots as $F$, but each root
occurs once. Classification signs are always evaluated using the original
$F$, since removing repeated factors can change those signs.

We refine the isolating intervals enough to determine which neighboring
integers can lie across a root, checking a candidate integer root by exact
evaluation. The signs of the original $F$ at these integers determine the
group boundaries without examining all indices in $[N]$. Appendix~\ref{app:root-to-grid} gives an elementary separation bound that
justifies the required refinement and handles exact integer roots separately.
Thus the preprocessing cost is polynomial in $d$, $\beta$, and $n$.
No root finding is performed in superposition.

Once the group boundaries are known, we can classify an index and convert
between an index and its rank using comparisons with stored constants. Let
$j\in\{1,\ldots,G\}$ satisfy $\theta_j=\mathcal{L}$. For any
$i\in\mathcal{I}_{\mathcal{L}}$ satisfying
$\gamma_{j-1}\le i<\gamma_j$, its rank is
\begin{equation}
 \operatorname{rank}_{\mathcal{I}_{\mathcal{L}}}(i)=
 \lambda_j+i-\gamma_{j-1}+1.
\label{eq:light-rank}
\end{equation}
If $j\in\{1,\ldots,G\}$ satisfies $\theta_j=\mathcal{H}$ and
$i\in\mathcal{I}_{\mathcal{H}}$ satisfies
$\gamma_{j-1}\le i<\gamma_j$, the same expression applies with $\lambda_j$
replaced by $\eta_j$.
Conversely, given a rank, we identify its group from the stored
values of $\lambda_j$ and $\eta_j$ and recover the corresponding index by
addition. Since there are $O(d)$ groups, these operations require $O(d)$
comparisons and controlled additions.

\subsection{Exact cumulative functions}

For every nonnegative integer $t$, define the cumulative value of $Q=P^2$ by
\begin{equation}
 S(t):=\sum_{i=0}^{t-1}Q(i),
 \qquad S(0)=0.
\label{eq:S-def}
\end{equation}
We will only evaluate $S(t)$ for $0\le t\le N$. Although the sum in
\cref{eq:S-def} may contain $N$ terms, its value has an exact representation
using the $2d+1$ integers introduced below.
For every nonnegative integer $a$ and every integer $k$, we use the convention
$\binom{a}{k}=0$ when $k<0$ or $k>a$.

\begin{proposition}[Exact cumulative sum]
\label[proposition]{prop:exact-cumulative}
For a function $g$ on the nonnegative integers, define its forward difference
for nonnegative integers $u$ and $k$ by
\begin{equation}
 \begin{aligned}
 (\Delta g)(u)&:=g(u+1)-g(u),
 &\Delta^0 g&:=g,\\
 \Delta^{k+1}g&:=\Delta(\Delta^k g).
 \end{aligned}
\label{eq:forward-difference}
\end{equation}
For $k\in\{0,\ldots,2d\}$, set
$\kappa_k:=(\Delta^k Q)(0)$. Then, for all nonnegative integers $i$ and $t$,
\begin{equation}
 Q(i)=\sum_{k=0}^{2d}\kappa_k\binom{i}{k},
 \qquad
 S(t)=\sum_{k=0}^{2d}\kappa_k\binom{t}{k+1}.
\label{eq:S-binomial}
\end{equation}
All coefficients $\kappa_k$ are integers.
\end{proposition}

\begin{proof}
Since $Q$ has degree at most $2d$, its $(2d+1)$st forward difference is zero.
Repeatedly applying $g(u+1)=g(u)+(\Delta g)(u)$ and the binomial identity
$\binom{i+1}{k}=\binom{i}{k}+\binom{i}{k-1}$ proves the first equality by
induction on $i$.
Summing it from $i=0$ to $t-1$ and using
$\sum_{i=0}^{t-1}\binom{i}{k}=\binom{t}{k+1}$ gives the second equality.
Moreover, $Q(u)$ is an integer for every integer $u$, so repeated subtraction
in \cref{eq:forward-difference} shows that every $\kappa_k$ is an integer.
Equivalently,
\begin{equation}
 \kappa_k=\sum_{u=0}^{k}(-1)^{k-u}\binom{k}{u}Q(u)\in\mathbb Z.
\label{eq:finite-difference-coefficients}
\end{equation}
\end{proof}
The coefficients $\kappa_k$ are computed classically during preprocessing.
When evaluating $S(t)$ coherently, we generate the binomial coefficients
needed in \cref{eq:S-binomial} reversibly from the register holding $t$,
which may be in superposition. We generate them in increasing order using
the following recurrence for $k\in\{0,\ldots,2d\}$:
\begin{equation}
 \binom{t}{0}=1,
 \qquad
 \binom{t}{k+1}=
 \begin{cases}
  \displaystyle\binom{t}{k}\frac{t-k}{k+1},&k<t,\\[4pt]
 0,&k\ge t.
 \end{cases}
\label{eq:binomial-recurrence}
\end{equation}
If $k<t$, the quotient in
\cref{eq:binomial-recurrence} is an integer because
$(k+1)\binom{t}{k+1}=(t-k)\binom{t}{k}$. If
$k\ge t$, the next coefficient and all subsequent coefficients are zero.
The division is applied to the product
$\binom{t}{k}(t-k)$, not to $t-k$ alone. Since $k+1$ is a known positive
integer at each step, every such division has an integer result whenever
$k<t$. The branch $k\ge t$ is handled directly and does not attempt
to evaluate a negative factor or a quotient outside its specified domain.
The coefficients may be generated in $O(d)$ ancilla registers. After the sum
$S(t)$ has been copied to its output register, the calculation is reversed so
that each ancilla qubit returns to $\ket0$. Thus $S(t)$ can be evaluated with
$O(d)$
multiplications, divisions by known positive integers, additions, and
comparisons.

Appendix~\ref{subsec:constant-division} implements these exact divisions
reversibly using shifts and multiplication by precomputed modular inverses
of the odd parts of the divisors. This construction justifies the division
steps in the recurrence.

We now express $D$ and $E$ using $S$. Let
$j\in\{1,\ldots,G\}$ satisfy $\theta_j=\mathcal{L}$, and let
$r\in\{1,\ldots,m_{\mathcal{L}}\}$ satisfy
$\gamma_{j-1}\le\ell_r<\gamma_j$. Equivalently,
\begin{equation}
 \lambda_j<r\le\lambda_j+\gamma_j-\gamma_{j-1}.
\label{eq:light-rank-range}
\end{equation}
The number of light indices in this group up to and including $\ell_r$ is
$t:=r-\lambda_j$, so $1\le t\le\gamma_j-\gamma_{j-1}$. Using the stored value
$D(\lambda_j)$, we have
\begin{equation}
 D(r)=D(\lambda_j)+tZ
 -N\bigl[S(\gamma_{j-1}+t)-S(\gamma_{j-1})\bigr].
\label{eq:D-closed}
\end{equation}
Similarly, let $j\in\{1,\ldots,G\}$ satisfy $\theta_j=\mathcal{H}$, and let
$s\in\{1,\ldots,m_{\mathcal{H}}\}$ satisfy
$\gamma_{j-1}\le h_s<\gamma_j$. Equivalently,
\begin{equation}
 \eta_j<s\le\eta_j+\gamma_j-\gamma_{j-1}.
\label{eq:heavy-rank-range}
\end{equation}
Setting $t:=s-\eta_j$, so that
$1\le t\le\gamma_j-\gamma_{j-1}$, and using the stored value $E(\eta_j)$, we
obtain
\begin{equation}
 E(s)=E(\eta_j)
 +N\bigl[S(\gamma_{j-1}+t)-S(\gamma_{j-1})\bigr]-tZ.
\label{eq:E-closed}
\end{equation}
The values $D(0)=E(0)=0$ are evaluated directly. For
$1\le r\le m_{\mathcal{L}}$, the circuit selects the light group containing
$\ell_r$ using \cref{eq:light-rank-range}; for
$1\le s\le m_{\mathcal{H}}$, it selects the heavy group containing $h_s$
using \cref{eq:heavy-rank-range}.

The circuit checks these conditions against each of the $O(d)$ precomputed
group records. Each record contains the group's type, boundaries, counts,
and cumulative weights. These data are computed classically and incorporated
as fixed constants into the circuit. Comparisons identify the matching
group, and controlled additions select its data for use in
\cref{eq:D-closed,eq:E-closed}. This selection is coherent when the input
rank is in superposition. Consequently, evaluating $D(r)$ or $E(s)$ requires
$O(d)$ comparisons and controlled additions, including the selection of the
stored data, and two evaluations of $S$.
Appendix~\ref{subsec:group-selection} gives the explicit selection procedure. 

\subsection{Arithmetic width and gate conventions}

Bit complexity matters because the exact integers can be much wider than
an $n$-bit grid index. Choose a sufficient magnitude width
\begin{equation}
 B=O\bigl(\beta+dn+d\log(d+1)\bigr)
 \label{eq:arithmetic-width}
\end{equation}
for the absolute values of $P,Q,S,D,E$ and their intermediate arithmetic, with
$B\ge n+1$ and $NZ<2^{B}$.
A signed intermediate uses an additional sign bit. The binomial calculation
is exact: the width must bound its signed intermediate terms, not only the
final nonnegative sum. Appendix~\ref{app:arithmetic-width-proof} derives an
explicit sufficient choice in \cref{eq:explicit-arithmetic-width}.
The circuit for preparing a uniform state is allowed an extra endpoint bit for the constant
$2^{B}$.

We count elementary gates over Clifford+$T$, so reversible Boolean
arithmetic can be implemented exactly using NOT, CNOT, and Toffoli gates,
with a decomposition of each Toffoli into a constant number of gates. This is a total gate
count, not a concrete count of $T$ gates or an estimate of circuit depth.
Let $\Mcost(L)\ge L$ and $\mathsf W(L)$ bound, respectively,
the numbers of elementary gates and ancilla qubits in a chosen clean
reversible multiplication circuit on $L$-bit magnitudes.
The circuit preserves its two $L$-qubit input registers and XORs their
full product into a $2L$-qubit output register. The bound $\mathsf W(L)$
counts internal ancilla qubits beyond these input and output registers.
The gate count includes internal uncomputation. The ancilla qubits are
initialized to $\ket0$ and returned to $\ket0$. We take these resource bounds to be nondecreasing
in $L$. Additions and comparisons cost $O(L)$ gates and are therefore
bounded by $O(\Mcost(L))$. Exact division by a known constant
$h=2^\nu o$, with $o$ odd, reduces to a shift by $\nu$ bits followed by
multiplication by the precomputed inverse of $o$ modulo $2^L$.
Appendix~\ref{subsec:constant-division} gives the reversible implementation,
which uses $O(\Mcost(L))$ gates. 

For the quasi-linear estimates below, we use fast integer multiplication
together with reversible compute--copy--uncompute.
Harvey and van der Hoeven~\cite{harvey2021integer} establish the classical
$O(L\log L)$ multiplication bound. The reversible implementation described
in Appendices~A.1 and~A.4 of
Ragavan and Vaikuntanathan~\cite{ragavan2023space} gives
\begin{equation}
 \begin{aligned}
  \Mcost(L)&=O\!\left(L\log(2L)\right)=\Ot(L),\\
  \mathsf W(L)&=O\!\left(L\log(2L)\right)=\Ot(L).
 \end{aligned}
 \label{eq:quasilinear-multiplication}
\end{equation}
The circuit computes the product, XORs it into the output register,
and reverses the intermediate computation. It therefore acts coherently
on superpositions and returns all ancilla qubits to zero.
This implementation uses $O(L\log(2L))$ ancilla qubits.

We retain $\Mcost$ and $\mathsf W$ to display the dependence on the
arithmetic implementation. Schoolbook multiplication permits
$\Mcost(L)=O(L^2)$ and $\mathsf W(L)=O(L)$, while reversible Karatsuba
multiplication permits $\Mcost(L)=O(L^{\log_2 3})$ and
$\mathsf W(L)=O(L)$~\cite{gidney2019}.

Unless stated otherwise, the oracle for a function $g$ preserves its input
$x$, applies bitwise XOR of the binary encoding of $g(x)$ to a target
register holding an arbitrary bit string $b$, and returns each ancilla
qubit to $\ket0$:
\begin{equation}
 \begin{aligned}
 &O_g\ket{x}\ket{b}\ket{0^{a_g}}_{\rm ancilla}\\
 &\quad=\ket{x}\ket{b\mathbin{\oplus}g(x)}\ket{0^{a_g}}_{\rm ancilla}.
 \end{aligned}
 \label{eq:xor-oracle-convention}
\end{equation}
Here $b$ and the encoding of $g(x)$ have the same width, $\oplus$ denotes
bitwise XOR, and $a_g$ is the number of ancilla qubits used by the oracle.
We call this evaluation \emph{clean} because those ancilla qubits begin and
end in $\ket{0^{a_g}}$. It is implemented by computing $g(x)$ into an
ancilla register, XORing that value into the target register, and reversing
the computation. Its gate count includes these three steps. Values outside a stated
valid domain are assigned a deterministic default, with arithmetic operations
enabled only on inputs in their specified domains,
so the circuit is a unitary on all computational basis states. No claim
about those default values is needed on the valid input support.

Separately, let $\mathsf S_{\rm rot}(\eta)$ denote the \emph{number of gates in}
a compiled controlled rotation whose error in operator norm is at most $\eta$.
It is not the classical runtime of finding the gate sequence. We assume
$\mathsf S_{\rm rot}(\eta)=\operatorname{poly}(\log(1/\eta))$~\cite{dawsonnielsen2005}.
The explicit precision bounds in
\Cref{cor:precision-regimes} use
$\mathsf S_{\rm rot}(\eta)=O(\log(2/\eta))$, as established in
Appendix~\ref{app:controlled-rotation-synthesis}.
Here $\eta$ is the error tolerance for one rotation, whereas $\epsilon$
denotes the total error tolerance for state preparation.

Angles that depend on quantum registers are computed using reversible
fixed-point arithmetic. Here \emph{fixed-point} means that a fixed number
$p$ of binary digits represents the fractional part of a number, so the
binary point has a fixed position. A register storing an integer $A$ is
interpreted as the value $A/2^p$. For example, with $p=3$, the stored integer
$11$ represents $11/8=1.375$. Calculations operate on these integer
encodings, with results rounded to the prescribed precision. To implement
them reversibly, the circuit retains the information needed to undo each
operation. After the computed angle controls the rotations, uncomputation
returns each ancilla qubit to $\ket0$.

These angle computations and the synthesized
rotations have costs separate from $\Mcost(B)$.
Section~\ref{subsec:uniform-implementation} gives their gate bounds;
Appendices~\ref{app:controlled-synthesis} and~\ref{app:angle-computation}
detail the synthesis and angle computation, respectively. 
Both costs are included in the total preparation cost in
\cref{eq:clean-total-cost}.

Consequently,
\cref{eq:S-binomial,eq:D-closed,eq:E-closed} give the reversible maps
\begin{equation}
 \ket{r}\ket{0}\longmapsto\ket{r}\ket{D(r)},
 \qquad
 0\le r\le m_{\mathcal{L}},
 \label{eq:D-oracle}
\end{equation}
and
\begin{equation}
 \ket{s}\ket{0}\longmapsto\ket{s}\ket{E(s)},
 \qquad
 0\le s\le m_{\mathcal{H}}.
\label{eq:E-oracle}
\end{equation}
For either map, let $C_{\rm cum}$ bound the gate count for one evaluation
with a control qubit when needed, including the uncomputation that returns its
ancilla qubits to zero. The inverse circuit has the same gate count.
Selecting the group and evaluating $S$ twice require $O(d)$ arithmetic
operations, each costing at most $O(\Mcost(B))$ gates. Adding a control
changes the cost of each reversible Boolean gate by at most a constant
factor. Thus
\begin{equation}
 C_{\rm cum}=O\bigl(d\Mcost(B)\bigr).
 \label{eq:cumulative-cost}
\end{equation}

\subsection{Reversible monotone searches}

The minima in \cref{eq:s-of-r,eq:rho-of-s} can be found by a reversible binary
search. It has a fixed $n$-step schedule; its decisions are quantum control
bits, not measurements. For a fixed $r\in\{1,\ldots,m_{\mathcal{L}}\}$, we search over a heavy
rank $a\in\{1,\ldots,m_{\mathcal{H}}\}$ for the first $a$ satisfying
$E(a)>D(r-1)$; the result is $s(r)$. For a fixed
$s\in\{1,\ldots,m_{\mathcal{H}}\}$, we search over a light rank
$u\in\{1,\ldots,m_{\mathcal{L}}\}$ for the first $u$ satisfying
$D(u)\ge E(s)$; the result is $\rho(s)$. Because every $\delta_r$ and $e_s$ is
positive, $D$ and $E$ are strictly increasing, so each condition is false
before its minimum and true from its minimum onward. For the first search,
$E(0)\le D(r-1)$ and $E(m_{\mathcal{H}})>D(r-1)$. For the second search,
$D(0)<E(s)$ and $D(m_{\mathcal{L}})\ge E(s)$. These endpoint inequalities
follow from \cref{eq:DE-conservation}.

Both searches occur only when there is at least one light and one heavy index.
Hence $m_{\mathcal{L}},m_{\mathcal{H}}\le N-1$, so $n$ iterations suffice.
For either search, begin with $x=0$ and consider the $n$ increments
$2^{n-1},\ldots,2,1$ in decreasing order. At the iteration with exponent
$\nu\in\{n-1,\ldots,0\}$, set $x':=x+2^\nu$. For the search for $s(r)$,
replace $x$ by $x'$ precisely when $x'\le m_{\mathcal{H}}$ and
$E(x')\le D(r-1)$. For the search for $\rho(s)$, replace $x$ by $x'$ precisely
when $x'\le m_{\mathcal{L}}$ and $D(x')<E(s)$. After the $n$ iterations, $x$
is the largest rank at which the relevant condition is false, so the required
rank is $x+1$.

At each step, the circuit checks the condition for replacing $x$ by $x'$
and stores the result in an ancilla qubit: $1$ if the condition holds and
$0$ otherwise. It then uncomputes the temporary cumulative value and uses
this qubit to control the update of $x$.
Thus the circuit reuses the $O(dB+\mathsf W(B))$ ancilla qubits for
arithmetic at every step. It retains the comparison result from each
iteration in a separate ancilla qubit. It then copies $x+1$ to the output and applies the
iterations in reverse, returning $x$ and these $n$ qubits to zero. The
comparisons with $m_{\mathcal{H}}$ and
$m_{\mathcal{L}}$ control the evaluations of $E(x')$ and $D(x')$,
respectively, so neither function is evaluated outside the domains stated in
\cref{eq:D-oracle,eq:E-oracle}. On reversal, undo the update of $x$ \emph{before}
recomputing the comparison result to reset the corresponding ancilla qubit. This restores
the same inputs to the comparisons as in the forward step.

The cost below includes copying the search result and reversing the search
to reset its internal ancilla qubits to zero. The \emph{search threshold}
is the value against which the cumulative values are compared:
$D(r-1)$ when finding $s(r)$ in \cref{eq:s-of-r}, and $E(s)$ when finding
$\rho(s)$ in \cref{eq:rho-of-s}. For this cost bound, the threshold is
already available in an input register, which the search leaves unchanged.

We use the bound $C_{\rm cum}$ from \cref{eq:cumulative-cost}.
In each forward iteration,
one call to the cumulative oracle computes $D(x')$ or $E(x')$ into its
output register and resets its internal ancillas. After the comparison
result has been stored in its ancilla qubit, a second call, to the inverse
oracle, erases the cumulative value. These calls cost at most $2C_{\rm cum}$.
Reversing the iteration also requires two calls: recompute the cumulative
value to reset that ancilla qubit, then erase that value again. The $n$
forward iterations and $n$ reverse iterations therefore require at most
$4nC_{\rm cum}$ gates for cumulative evaluations.

The remaining comparisons, operations on the corresponding ancilla qubit, and updates
of the candidate rank cost $O(B)$ gates per iteration, since $B\ge n+1$.
Forming and copying the final rank $x+1$ cost $O(n)$ gates. Including these
operations in both directions gives $O(nB)$ additional gates. Consequently,
\begin{equation}
 C_{\rm search}\le4nC_{\rm cum}+O(nB)
 =O\bigl(nd\Mcost(B)\bigr).
\label{eq:search-cost-detail}
\end{equation}
The last equality substitutes \cref{eq:cumulative-cost} and uses
$d\ge1$ and $\Mcost(B)\ge B$ to absorb the $O(nB)$ term.
At most $n$ ancilla qubits are needed to retain the comparison results;
the ancilla qubits used for arithmetic are reused at every step.
Computing the threshold from the input rank and later erasing it
are accounted for by the calling oracle.

\subsection{The implicit alias oracle and its cost}

Combining classification, rank conversion, evaluations of $D$ and $E$, and the
two searches gives the following oracle.

\begin{definition}[Polynomial implicit alias oracle]
\label[definition]{def:implicit-oracle}
The polynomial implicit alias oracle $\PIA$ is a reversible implementation of
\begin{equation}
 \PIA\ket{i}\ket{0}\ket{0}
 =\ket{i}\ket{\alias(i)}\ket{\keep(i)},
 \qquad i\in[N],
\label{eq:PIA-action}
\end{equation}
where $\alias$ and $\keep$ are defined in \cref{eq:ordered-table}.
The register for $\alias(i)$ has $n$ qubits, and the register for $\keep(i)$
has $\lceil\log_2(Z+1)\rceil$ qubits.
\end{definition}

Given $i$, the oracle first compares it with the stored group boundaries and
writes the corresponding group number $j$ and type $\theta_j$ into temporary
registers. If $\theta_j=\mathcal{L}$, it computes the light rank $r$ using
\cref{eq:light-rank}, evaluates $s(r)$, and converts that heavy rank back to the
index $h_{s(r)}$. If $\theta_j=\mathcal{H}$, it computes the heavy rank $s$ and
evaluates $\rho(s)$. The entries in these two cases are then obtained from
\cref{eq:ordered-light,eq:ordered-heavy-keep,eq:ordered-heavy}. If
$\theta_j=\mathcal{B}$, the oracle uses $\alias(i)=i$ and $\keep(i)=Z$ directly.
After copying $\alias(i)$ and $\keep(i)$ to their output registers, the oracle
reverses every intermediate calculation. The gates for all three cases form a
single circuit controlled by $\theta_j$, so the same circuit applies when the
input register is in a superposition of indices.

\figrev{Figure~\ref{fig:opia} summarizes this evaluation as a
compute--copy--uncompute procedure. Its two branches use the same compact
group data and cumulative functions; the balanced case requires no
monotone search. The input index is preserved, and the alias and threshold
remain in their output registers after the temporary group labels, ranks,
and search data have been uncomputed.}

\begin{figure*}[!t]
\centering
\includegraphics[width=0.6\textwidth]{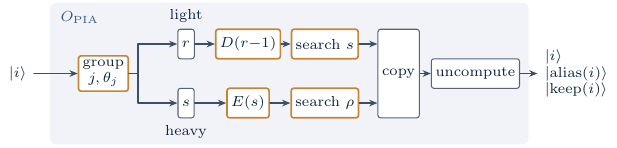}
\caption{Coherent evaluation of the polynomial implicit alias oracle
$\PIA$. Group selection determines the group $j$ and type $\theta_j$.
For a light index, its rank $r$ gives $D(r-1)$ and the search for $s(r)$;
for a heavy index, its rank $s$ gives $E(s)$ and the search for $\rho(s)$.
The resulting table entries are copied to the output registers before
all temporary computations are reversed. Amber outlines mark operations
using the fixed classical data. The branches are controlled coherently,
without measurement. The balanced case, which directly returns
$\alias(i)=i$ and $\keep(i)=Z$, is omitted from the schematic, as are
the final arithmetic and rank-to-index conversions.}
\label{fig:opia}
\end{figure*}

\begin{theorem}[Implicit representation]
\label[theorem]{thm:implicit-representation}
Let $P$ be the integer polynomial in \cref{eq:integer-P}, evaluated at the
indices in $[N]$, where $N=2^n$. Let $\beta$ be defined by
\cref{eq:coefficient-width}, and choose $B$ as described after
\cref{eq:arithmetic-width}. In addition to the fixed coefficients of $P$, the
ordered alias table has a classical representation consisting of $O(d)$
group records and the $2d+1$ coefficients $\kappa_k$. For
$j\in\{1,\ldots,G\}$, the $j$th group record contains the boundaries
$\gamma_{j-1}$ and $\gamma_j$ from \cref{eq:group-boundaries}; its light,
heavy, or balanced type $\theta_j$; the counts $\lambda_j$ and $\eta_j$
of light and heavy indices before $\gamma_{j-1}$, respectively, defined in
\cref{eq:group-earlier-counts}; and the corresponding cumulative missing
and extra weights $D(\lambda_j)$ and $E(\eta_j)$ from \cref{eq:DE}.
The representation also stores
$G$, $Z$, $m_{\mathcal{L}}$, and $m_{\mathcal{H}}$. These data can be computed
using a number of bit operations polynomial in $d$, $n$, and $B$, and
they occupy $O(d(n+B))$ bits. Here
$\Mcost(B)$ is the total gate count for multiplying two integers whose
absolute values can each be represented using at most $B$ bits.
By the arithmetic convention above, it also bounds the cost of each exact
arithmetic primitive used by the oracle up to a constant factor.
Writing $C_{\mathrm{gate}}(\PIA)$ for the total gate count of
one application of the oracle, we have
\begin{equation}
 C_{\mathrm{gate}}(\PIA)=O\bigl(nd\Mcost(B)\bigr)
\label{eq:PIA-cost}
\end{equation}
which specializes to $\Ot\bigl(ndB\bigr)$ for the
implementation in \cref{eq:quasilinear-multiplication}.
The construction uses
$O(dB+n+\mathsf W(B))$ ancilla qubits.
For the same implementation, this is
$\Ot(dB+n)$.
The $O(dB)$ term includes the $O(B)$ qubits used to hold the shifted
operand, fixed multiplier, and full product for exact division in
Appendix~\ref{subsec:constant-division}.
The ancilla qubits used for multiplication are reused between arithmetic operations.
The circuit does not store or query an explicit list of all $N$
pairs $(\alias(i),\keep(i))$.
\end{theorem}

\begin{proof}
By \Cref{prop:exact-cumulative}, either map in
\cref{eq:D-oracle,eq:E-oracle} has gate complexity
$O(d\Mcost(B))$. Each binary search evaluates one of these maps in each of
$n$ forward steps and again while reversing those steps. Thus one search uses
$O(n)$ evaluations of $D$ or $E$. The controlled light and heavy cases require
only a constant number of such searches. Evaluating $v_i=NQ(i)$, forming a
value of $\keep(i)$, identifying a group, and converting between an index and
its rank require $O(d)$ additional arithmetic operations. Their costs are
absorbed by \cref{eq:PIA-cost}.

By \Cref{lem:classification-groups}, there are $O(d)$ group records. The
coefficients $\kappa_k$ are obtained by repeated subtraction, and
$Z=S(N)$ is then computed using \cref{eq:S-binomial}. If $F$ is nonzero, the
procedure for root isolation described above finds the group boundaries; if $F$ is
zero, there is one balanced group. The stored boundary values are then obtained
by one scan through the groups.
Initialize two accumulators $d_{\mathrm{cur}}=e_{\mathrm{cur}}=0$. Before
processing group $j$, store
$D(\lambda_j)=d_{\mathrm{cur}}$ and $E(\eta_j)=e_{\mathrm{cur}}$. For a light
group, update
\begin{equation}
 d_{\mathrm{cur}}\leftarrow d_{\mathrm{cur}}
 +(\gamma_j-\gamma_{j-1})Z
 -N\bigl[S(\gamma_j)-S(\gamma_{j-1})\bigr],
\label{eq:preprocess-D-update}
\end{equation}
and, for a heavy group, update
\begin{equation}
 e_{\mathrm{cur}}\leftarrow e_{\mathrm{cur}}
 +N\bigl[S(\gamma_j)-S(\gamma_{j-1})\bigr]
 -(\gamma_j-\gamma_{j-1})Z.
\label{eq:preprocess-E-update}
\end{equation}
A balanced group changes neither accumulator. Each of these calculations uses
a number of bit operations polynomial in $d$, $n$, and $B$. The stated storage
bound follows from the number and widths of the stored integers. If every index
is balanced, one balanced group describes the table and no binary search is
needed.
\end{proof}

The description and evaluation costs are polynomial in $n,d,\beta$.
Calling the description exponentially smaller than an explicit $N=2^n$ table
presumes, for example, $d,\beta=\operatorname{poly}(n)$; the theorem does
not claim compression for arbitrary exponentially large input polynomials.
It gives the cost of one application of the oracle and does not claim
that all $N$ entries of the alias table can be written down using fewer than $O(N)$
operations. When
the input is a superposition of indices, the circuit in \cref{eq:PIA-action}
evaluates the corresponding pair $(\alias(i),\keep(i))$ for every component of
the superposition.

\subsection{Classical preprocessing and coherent quantum operation}

The input polynomial is fixed before the quantum circuit is built.
The following preprocessing is classical and exact.
\begin{enumerate}
\item Form the monomial coefficients
\begin{equation}
 \alpha_k=a_kN^{d-k},\qquad P(i)=\sum_{k=0}^d\alpha_k i^k.
 \label{eq:monomial-coefficients}
\end{equation}
Evaluate $Q(u)$ at $u=0,\ldots,2d$, take finite differences to obtain the
$\kappa_k$, and compute
\begin{equation}
 Z=S(N)=\sum_{k=0}^{2d}\kappa_k\binom{N}{k+1}.
 \label{eq:normalization-from-differences}
\end{equation}
\item Form $F=NP^2-Z$, isolate its roots when it is nonzero, and construct the
groups of constant sign. A scan of those $O(d)$ groups gives their counts and
cumulative boundary values using
\cref{eq:preprocess-D-update,eq:preprocess-E-update}.
\item Compile these integers as constants in the arithmetic circuits. For a
chosen circuit error, precompute and synthesize the fixed rotations for
$U_Z$ and choose the precision of the circuit for a variable range.
\end{enumerate}
The first step uses $O(d)$ polynomial evaluations.
The compression is useful when $d$ and $\beta$ are small
relative to $N$; in particular, they may be polynomial in $n=\log_2N$.

During quantum execution, only the \emph{arguments} of the arithmetic are in
superposition. An input $i$ determines its group, its light or heavy rank,
and the cumulative values needed by a monotone search. Those calculations
produce $\keep(i)$ and $\alias(i)$; ranking and unranking reuse the same
operations. The output $y$ later determines the fiber volume $v_y=NP(y)^2$,
the sign of $P(y)$, and the rotation angles used to prepare and unprepare a
uniform state over the $v_y$ fiber positions.

The roots are used during classical preprocessing to compute the group
boundaries. The circuit uses these boundaries, the coefficients $\kappa_k$,
and $Z$ as fixed data. The gate bounds include the cost of selecting among
$O(d)$ stored group records. For fixed $P$, the requested circuit accuracy
determines the precision used to compute angles and synthesize rotations.
The pipeline is
\begin{equation}
 \begin{gathered}
 (n,a_0,\ldots,a_d)\xrightarrow{\rm classical}\text{fixed data}\\
 \xrightarrow{\rm compile}\text{quantum circuit}.
 \end{gathered}
 \label{eq:classical-quantum-pipeline}
\end{equation}

\section{State preparation and fiber ranking}
\label{sec:clean-preparation}

Recall from \Cref{def:fragment} that a fragment is a pair $(i,z)$:
$i\in[N]$ identifies an entry of the alias table, and $z\in[Z]$ labels one
of the $Z$ fragments initially assigned to that index. The shuffling map
$\pi$ in \cref{eq:fragment-map} sends this fragment to the output
$y=\pi(i,z)$, which is either $i$ or $\alias(i)$ according to whether
$z<\keep(i)$ or $z\ge\keep(i)$. Coherent evaluation writes $y$ into a new
register while retaining $i$ and $z$. On a uniform superposition of input
fragments, it gives
\begin{equation}
 \begin{gathered}
 \frac{1}{\sqrt{NZ}}\sum_{\substack{i\in[N]\\z\in[Z]}}
 \ket{0^n}\ket{i}\ket{z}\\
 \longmapsto
 \frac{1}{\sqrt{NZ}}\sum_{\substack{i\in[N]\\z\in[Z]}}
 \ket{\pi(i,z)}\ket{i}\ket{z}\\
 =\sum_{y:v_y>0}\frac{|P(y)|}{\sqrt Z}
 \ket{y}\ket{\mathrm{temp}_y}.
 \end{gathered}
 \label{eq:alias-register-overview}
\end{equation}
The three registers on the first two lines hold, in order, the output index
$y$, the source index $i$, and the position $z$ within the fragments from
source $i$. The last line is the regrouping in \cref{eq:fragment-state}:
$\ket{\mathrm{temp}_y}$ is the uniform superposition of source pairs
$(i,z)$ satisfying $\pi(i,z)=y$. By \cref{eq:fiber-size}, measuring the
output gives $y$ with probability $v_y/(NZ)=P(y)^2/Z$. Distinct outputs
have disjoint fibers, so the states $\ket{\mathrm{temp}_y}$ are orthogonal.
Thus the output register is entangled with the pair of retained registers
whenever at least two outputs have nonzero probability.

We combine fiber ranking with circuits for preparing uniform states to
remove this entanglement. Applying the phase determined by the sign of
$P(y)$ then gives the pure state in \cref{eq:psi-P}.
We first introduce the two circuits for preparing
uniform states, then construct the ranking procedure. Their implementations
and full costs follow the ranking construction.

\begin{table}[!htbp]
\caption{Gate complexity of coherent subroutines for exact arithmetic.
Group selection identifies the classification group for an index or a
specified light or heavy rank and retrieves its stored data. A rank is the
position of an index in the ordered list of light or heavy indices.
Conversion in either direction includes selecting the appropriate group.
Each bound includes returning every ancilla qubit to $\ket0$ by
uncomputation. Implementing a constant number of queries to an oracle and
its inverse preserves the corresponding asymptotic gate bound.}
\label{tab:arithmetic-ledger}
\small
\begin{tabular}{@{}p{0.62\columnwidth}p{0.31\columnwidth}@{}}
\toprule
Subroutine & Gate Complexity \\
\midrule
Selecting a classification group & $O(dB)$ \\
Converting between an index and its rank & $O(dB)$ \\
$P,Q,v$; $S$; either $D$ or $E$ & $O(d\Mcost(B))$ \\
One monotone search & $O(nd\Mcost(B))$ \\
$\PIA$, $R$, or $R^{-1}$ & $O(nd\Mcost(B))$ \\
\bottomrule
\end{tabular}
\end{table}

\subsection{Two different primitives for preparing uniform states}

For $1\le m<2^{B}$, choose an ideal unitary $U_m$ on a
$B$-qubit register such that
\begin{equation}
 U_m\ket{0^{B}}
 =\frac{1}{\sqrt m}\sum_{t=0}^{m-1}\ket t,
 \label{eq:Um}
\end{equation}
and define $U_0$ to be the identity. The construction below uses a binary tree
to specify a unitary extension to every basis state.

The two circuits for preparing uniform states act at different stages of
the preparation.
The circuit $U_Z$, with \emph{fixed capacity}, prepares the initial uniform
superposition of fragment positions, and $Z$ is a classical constant.
After ranking, the circuit $U_{v_y}^{\dagger}$, with \emph{variable fiber volume},
returns the register holding the fiber position to zero. Its range $v_y$ depends on the output
index $y$, so the circuit must select the appropriate inverse for every $y$
in superposition, without measuring $y$. For fixed capacity, the rotation
angles are determined classically; for variable fiber volume, the circuit
requires coherent arithmetic unless a different implementation for variable
ranges is supplied.

\subsection{Coherent alias sampling}

Apply $\Had^{\otimes n}$ to an index register and $U_Z$ to a register holding
the fragment position. This prepares the uniform superposition
\begin{equation}
 \frac{1}{\sqrt{NZ}}
  \sum_{i\in[N]}\sum_{z\in[Z]}\ket{i}\ket{z}.
\label{eq:uniform-fragments}
\end{equation}
Append an $n$-qubit output register initialized to zero. One call to $\PIA$
computes $\alias(i)$ and $\keep(i)$, and a comparison bit $\chi$ is set to $1$
if and only if $z\ge\keep(i)$. The circuit writes $i$ to the output register
when $\chi=0$ and writes $\alias(i)$ when $\chi=1$. Its value is therefore
$\pi(i,z)$. The comparison is reversed and one call to $\PIA^\dagger$ clears the two
outputs of the alias table. Thus this sampling subroutine uses \emph{one forward
and one inverse} call to the alias oracle. Its exact gate cost
is $2C_{\rm gate}(\PIA)+O(n+B)$. The resulting state is
\begin{equation}
 \frac{1}{\sqrt{NZ}}
 \sum_{i\in[N]}\sum_{z\in[Z]}
 \ket{\pi(i,z)}\ket{i}\ket{z}.
\label{eq:alias-with-garbage}
\end{equation}
Regrouping the fragments by their output index and using
\cref{eq:fiber-size} gives
\begin{equation}
 \sum_{\substack{y\in[N]\\v_y>0}}
 \frac{|P(y)|}{\sqrt Z}\ket{y}\ket{\mathrm{temp}_y},
\label{eq:garbage-state}
\end{equation}
where
\begin{equation}
 \ket{\mathrm{temp}_y}
 =\frac{1}{\sqrt{v_y}}
  \sum_{(i,z)\in\pi^{-1}(y)}\ket{i}\ket{z}.
\label{eq:temporary-state}
\end{equation}
The two registers holding the source index $i$ and fragment position $z$
therefore retain a garbage state $\ket{\mathrm{temp}_y}$ that depends on
the output $y$. If these registers are retained, a phase computed from the
sign of $P(y)$ may be applied at this point to restore the sign of each
amplitude. In the garbage-free construction below, this phase is postponed
until the registers encoding the fragment have been returned to zero.

The states $\ket{\mathrm{temp}_y}$ have disjoint supports in the computational basis and
are orthogonal for distinct output indices: each fragment $(i,z)$ produces
exactly one output $y$. Discarding the registers holding $(i,z)$ therefore
removes the coherence between distinct outputs and produces the diagonal
mixture in \cref{eq:alias-reduced-state}. The state
in \cref{eq:garbage-state} is therefore a preparation with retained garbage;
a standalone preparation of \cref{eq:psi-P} requires the registers holding the source
index and fragment position to be erased.

\subsection{Fiber ranking}

\begin{definition}[Efficient fiber ranking]
\label[definition]{def:fiber-ranking}
Let $f:X\to Y$ be a map of finite sets. For every $y\in f(X)$, define
$m_y:=|f^{-1}(y)|$. The preimage $f^{-1}(y)$ is called
the \emph{fiber} of $y$, and its volume is $m_y$. A
\emph{fiber ranking} is a bijection
\begin{equation}
 R_f:X\longrightarrow
 \{(y,t):y\in f(X),\ 0\le t<m_y\}
\label{eq:ranking-map}
\end{equation}
whose first component is $f$, so $R_f(x)=(f(x),t)$. The position $t$
distinguishes inputs that produce the same output $y$, so the pair $(y,t)$
uniquely determines the original input $x$. Positions within a fiber are
labeled $0,1,\ldots,m_y-1$, whereas the light and heavy ranks $r$ and $s$
start from $1$ by \cref{eq:index-ranks}. The
ranking is \textit{efficient} if $R_f$, $R_f^{-1}$, and the map $y\mapsto m_y$ have
reversible implementations that return each ancilla qubit to $\ket0$
and use a number of gates polynomial in $b_f$. Here $b_f$ is the sum of
the maximum numbers of bits needed to encode one element of $X$ and one
element of $Y$, and the number of bits in the fixed data specifying $f$.
\end{definition}

Every $y\in f(X)$ has a nonempty fiber, so $m_y\ge1$.

When $m_y<2^{B}$, apply $U_{m_y}^{\dagger}$ to the position register
under the control of a register holding the fiber volume $m_y$. The joint
unitary is
\begin{equation}
 \mathcal U^{\dagger}:=
 \sum_{m=0}^{2^{B}-1}\ket{m}\!\bra{m}\otimes U_m^{\dagger}.
\label{eq:selected-uniform-inverse}
\end{equation}
The first register stores $m_y$, and the second is the $B$-qubit position
register. For each $m$, this operation maps
$\ket m\ket\phi$ to $\ket m U_m^\dagger\ket\phi$:
it preserves the register holding $m$ and applies the corresponding inverse
to the position register, coherently for all values of $m$ in superposition.
In particular, when $m=m_y$, it maps the uniform position state
$U_{m_y}\ket{0^B}$ to $\ket{0^B}$.

\figrev{Figure~\ref{fig:fiber-garbage} illustrates why this erasure is
possible for the running example. The $288$ source fragments producing
$y=6$ occupy three separated regions of the alias table. Ranking assigns
them consecutive positions $0,\ldots,287$, so their conditional state
becomes exactly $U_{288}\ket0$ after the source registers are cleared
by inverse ranking. Applying $U_{288}^{\dagger}$ then removes the
position register's dependence on the output, while preserving the
output amplitude $\sqrt{288/(8\cdot140)}=6/\sqrt{140}$.}

\begin{figure*}[!t]
\centering
\includegraphics[width=0.8\textwidth]{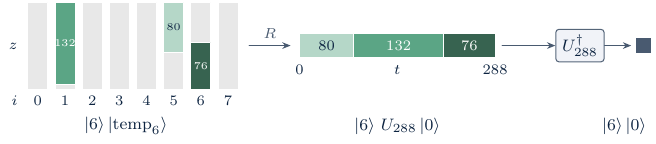}
\caption{Fiber ranking and garbage removal for output $y=6$ in the
eight-index example. The highlighted source regions are
$(5,z)$ with $60\le z<140$, $(1,z)$ with $8\le z<140$, and
$(6,z)$ with $0\le z<76$. Ranking places their $80$, $132$, and $76$
fragments in the consecutive intervals $[0,80)$, $[80,212)$, and
$[212,288)$, respectively. The arrow labeled $R$ denotes the reversible
change of coordinates implemented by computing the ranking and then
using inverse ranking to clear the source registers. Conditioned on
$y=6$, the remaining position state is $U_{288}\ket0$; its inverse
returns this register to zero. The displayed states are normalized
conditional states, with the common amplitude $6/\sqrt{140}$ suppressed.
The circuit applies this procedure coherently for all nonempty fibers;
exact erasure here assumes ideal rotations.}
\label{fig:fiber-garbage}
\end{figure*}

\begin{proposition}[Uncomputation by ranking]
\label[proposition]{prop:ranking-principle}
Let $f:X\to Y$, with $X\ne\varnothing$, admit an efficient fiber ranking.
Suppose $m_y<2^{B}$ and that a uniform state on $X$ has been prepared.
Let $n_X$ and $n_Y$ be the widths of the registers encoding $X$ and $Y$.
The ancilla register consists of the two $B$-qubit registers holding the
fiber position and volume, together with the ancilla qubits used for
arithmetic; denote its total width by $a_f$.
With ideal rotations for preparing uniform states, a unitary circuit maps
\begin{equation}
 \begin{gathered}
 \frac{1}{\sqrt{|X|}}\sum_{x\in X}\ket x_{\mathsf X}\ket{0^{n_Y}}_{\mathsf Y}
 \ket{0^{a_f}}_{\rm ancilla}\\
 \longmapsto\\
 \sum_{y\in f(X)}\sqrt{\frac{m_y}{|X|}}
 \ket{0^{n_X}}_{\mathsf X}\ket y_{\mathsf Y}\ket{0^{a_f}}_{\rm ancilla}.
 \end{gathered}
 \label{eq:ranking-principle}
\end{equation}
Using oracles that act as in \cref{eq:xor-oracle-convention}, the circuit
uses one query each to $O_{R_f}$ and $O_{R_f^{-1}}$, one query each to the
oracle $O_m$ for computing the fiber volume and its inverse $O_m^\dagger$, and one application
of $\mathcal U^\dagger$ from \cref{eq:selected-uniform-inverse}.
\end{proposition}

\begin{proof}
Write $R_f(x)=(f(x),t_f(x))$. In addition to $\mathsf X$, the register holding
the original input, allocate registers $\mathsf Y$, $\mathsf T$, and $\mathsf M$
for the output, fiber position, and fiber volume. For a fixed $x$, put
$(y,t)=R_f(x)$. The first two oracle queries act as
\begin{equation}
 \begin{gathered}
 \ket x_{\mathsf X}\ket{0^{n_Y}}_{\mathsf Y}\ket{0^B}_{\mathsf T}\\
 \xrightarrow{\ O_{R_f}\ }
 \ket x_{\mathsf X}\ket{y}_{\mathsf Y}\ket{t}_{\mathsf T}\\
 \xrightarrow{\ O_{R_f^{-1}}\ }
 \ket{x\mathbin{\oplus}R_f^{-1}(y,t)}_{\mathsf X}
 \ket{y}_{\mathsf Y}\ket{t}_{\mathsf T}\\
 =\ket{0^{n_X}}_{\mathsf X}\ket{y}_{\mathsf Y}\ket{t}_{\mathsf T}.
 \end{gathered}
 \label{eq:ranking-basis-replacement}
\end{equation}
The oracle $O_{R_f^{-1}}$ evaluates the inverse function $R_f^{-1}$ on
$(y,t)$ and targets the original register $\mathsf X$. It computes
$R_f^{-1}(y,t)=x$ into an ancilla register, XORs that value into $\mathsf X$,
and reverses the computation to return each arithmetic ancilla qubit to
$\ket0$. The ranking is bijective, so $(y,t)$ uniquely determines $x$.

By linearity and bijectivity, the remaining registers contain
\begin{equation}
 \begin{aligned}
 &\frac{1}{\sqrt{|X|}}\sum_{x\in X}\ket{f(x)}_{\mathsf Y}\ket{t_f(x)}_{\mathsf T}\\
 &=\frac{1}{\sqrt{|X|}}\sum_{y\in f(X)}\sum_{t=0}^{m_y-1}\ket{y}_{\mathsf Y}\ket{t}_{\mathsf T}\\
 &=\sum_{y\in f(X)}\sqrt{\frac{m_y}{|X|}}
    \ket y_{\mathsf Y}\,U_{m_y}\ket{0^B}_{\mathsf T}.
 \end{aligned}
 \label{eq:ranked-uniform-state}
\end{equation}
This step relabels equally weighted basis states. Query $O_m$ to compute
$m_y$ into $\mathsf M$, apply $\mathcal U^\dagger$ from
\cref{eq:selected-uniform-inverse} to $\mathsf M$ and $\mathsf T$, and query
$O_m^\dagger$ to erase $m_y$. On each branch,
\begin{equation}
 \begin{gathered}
 \ket y_{\mathsf Y}\ket{0^B}_{\mathsf M}U_{m_y}\ket{0^B}_{\mathsf T}\\
 \xrightarrow{\ O_m\ }
 \ket y_{\mathsf Y}\ket{m_y}_{\mathsf M}U_{m_y}\ket{0^B}_{\mathsf T}\\
 \xrightarrow{\ \mathcal U^\dagger\ }
 \ket y_{\mathsf Y}\ket{m_y}_{\mathsf M}\ket{0^B}_{\mathsf T}\\
 \xrightarrow{\ O_m^\dagger\ }
 \ket y_{\mathsf Y}\ket{0^B}_{\mathsf M}\ket{0^B}_{\mathsf T}.
 \end{gathered}
 \label{eq:ranking-position-uncomputation}
\end{equation}
These operations leave the coefficient $\sqrt{m_y/|X|}$ in the output
register, proving the proposition. The application of $U_{m_y}^\dagger$,
controlled by $\mathsf M$, returns each conditional position state to
$\ket{0^B}$ while preserving $y$. Distinct output branches remain orthogonal because their
$y$ registers differ, so this cleanup is compatible with unitarity.

If $C_{R_f}$, $C_{R_f^{-1}}$, and $C_m$ are the gate counts for the oracles
in \cref{eq:xor-oracle-convention}, including their ancilla uncomputation,
the cost after preparing the uniform input is
\begin{equation}
 C_{R_f}+C_{R_f^{-1}}+2C_m+C_{\mathcal U^\dagger}.
 \label{eq:ranking-gate-complexity}
\end{equation}
The XOR into $\mathsf X$ and the internal uncomputation are already included
in $C_{R_f^{-1}}$. If the available evaluation interfaces require the target register to be zero,
one can instead compute $x$ into a separate register, XOR it into
$\mathsf X$, and reverse that evaluation. That alternative costs two
queries to the oracle evaluating $R_f^{-1}$ plus the explicit XOR and uses the extra copy
register. Both constructions have the same asymptotic cost.
\end{proof}

Every finite map has a fiber ranking, but it need not be efficient. For a
general alias table, finding the transferred fragments in $\pi^{-1}(y)$ may
require the set $\{i\in[N]:\alias(i)=y\}$. Storing these sets
would recreate an $N$-entry data structure. The ordered table avoids this
problem because \cref{eq:light-block-for-heavy} identifies the light indices
assigned to each heavy index by two rank boundaries.

\subsection{Ranking the ordered table}

We now construct the ranking for $\pi$. Extend $\rho$ by setting
$\rho(0):=0$, and set $\tau_0:=0$. If
$y\in\mathcal{I}_{\mathcal{L}}\cup\mathcal{I}_{\mathcal{B}}$ and $v_y>0$, then
\begin{equation}
 \pi^{-1}(y)=\{(y,z):0\le z<v_y\}.
\label{eq:light-balanced-preimage}
\end{equation}
For these fragments, define $t:=z$.

Now fix $s\in\{1,\ldots,m_{\mathcal{H}}\}$. The preimage of the heavy index
$h_s$ has three parts:
\begin{samepage}
\begin{enumerate}
\item $(h_{s-1},z)$ with
      $\keep(h_{s-1})\le z<Z$, when $s>1$;
\item $(\ell_r,z)$ with $\rho(s-1)<r\le\rho(s)$ and
      $v_{\ell_r}\le z<Z$;
\item $(h_s,z)$ with $0\le z<\keep(h_s)$.
\end{enumerate}
\end{samepage}
Order these fragments by listing the first part in increasing $z$, the second
part in increasing $r$ and then in increasing $z$ for each $r$, and the third
part in increasing $z$. The corresponding values of $t$ in the three parts are
\begin{subequations}
\label{eq:rank-formulas}
\begin{equation}
 t=z-\keep(h_{s-1}),
 \quad s>1,
 \quad \keep(h_{s-1})\le z<Z,
\label{eq:rank-previous-heavy}
\end{equation}
\begin{equation}
 \begin{gathered}
 t=\tau_{s-1}+D(r-1)-D(\rho(s-1))\\
 \phantom{t=}+z-v_{\ell_r},\\
 \rho(s-1)<r\le\rho(s),
 \qquad v_{\ell_r}\le z<Z.
 \end{gathered}
\label{eq:rank-light}
\end{equation}
Using $\tau_{s-1}=D(\rho(s-1))-E(s-1)$, the middle formula also reads
\begin{equation}
 t=D(r-1)-E(s-1)+z-v_{\ell_r}.
 \label{eq:rank-light-simplified}
\end{equation}
This form is useful in the circuit: once $s$ is known from the alias output,
it needs no additional search for $\rho(s-1)$.
For the third part, define
\begin{equation}
 t=v_{h_s}-\keep(h_s)+z,
 \quad 0\le z<\keep(h_s).
\label{eq:rank-own-heavy}
\end{equation}
\end{subequations}
The first case is absent when $s=1$.

\begin{proposition}[Ranking the ordered table]
\label[proposition]{prop:ordered-ranking}
The rules in \cref{eq:light-balanced-preimage,eq:rank-formulas} define a
bijection
\begin{equation}
 \begin{aligned}
 R:[N]\times[Z]&\longrightarrow
 \{(y,t):y\in[N],\ v_y>0,\ 0\le t<v_y\},\\
 (i,z)&\longmapsto(\pi(i,z),t).
 \end{aligned}
\label{eq:R-map}
\end{equation}
Both $R$ and $R^{-1}$ have oracle implementations of the form
\cref{eq:xor-oracle-convention}, with gate complexity
$O(nd\Mcost(B))$, the same asymptotic upper bound as $\PIA$.
Each implementation uses $O(dB+n+\mathsf W(B))$ additional qubits
for exact arithmetic.
\end{proposition}

\begin{proof}
For $h_s$, the three parts have sizes
\begin{equation}
 \tau_{s-1},
 \qquad D(\rho(s))-D(\rho(s-1)),
 \qquad \keep(h_s).
\label{eq:block-lengths}
\end{equation}
Using \cref{eq:tau,eq:DE,eq:ordered-heavy-keep}, the first two sizes sum to
\begin{align}
 &\tau_{s-1}+D(\rho(s))-D(\rho(s-1))\notag\\
 &\quad=D(\rho(s))-E(s-1)
       =v_{h_s}-\keep(h_s).
\label{eq:rank-boundary-identity}
\end{align}
Consequently, the values of $t$ in the three cases of
\cref{eq:rank-formulas} occupy, respectively,
\begin{equation}
 \begin{gathered}
 0\le t<\tau_{s-1},\\
 \tau_{s-1}\le t<v_{h_s}-\keep(h_s),\\
 v_{h_s}-\keep(h_s)\le t<v_{h_s}.
 \end{gathered}
\label{eq:heavy-rank-ranges}
\end{equation}
By \cref{eq:rank-boundary-identity}, these ranges meet at their endpoints and
together contain exactly $0,\ldots,v_{h_s}-1$, without repetition. For a light
or balanced index $y$, the rule $t=z$ maps the preimage in
\cref{eq:light-balanced-preimage} onto $\{0,\ldots,v_y-1\}$. This proves that
$R$ is a bijection.

For the forward map, first compute the alias data for $i$ and decide
whether the fragment is kept. A retained light or balanced fragment uses
$t=z$; a retained heavy fragment uses
$t=v_i-\keep(i)+z$. An aliased heavy fragment uses
$t=z-\keep(i)$. For an aliased light fragment, convert the output $y$ to its
heavy rank $s$ and use \cref{eq:rank-light-simplified}. Thus, beyond one
forward and one inverse $\PIA$ call needed to compute, copy, and clear the
alias data, only a constant number of $D,E,Q$ evaluations and rank
conversions is needed. In particular,
\begin{equation}
 C_R\le 2C_{\rm gate}(\PIA)+O(d\Mcost(B)).
 \label{eq:R-cost-detail}
\end{equation} For the inverse, if
$y\in\mathcal{I}_{\mathcal{H}}$, first compute
$s=\operatorname{rank}_{\mathcal{I}_{\mathcal{H}}}(y)$ using
\cref{eq:index-ranks}. Then compare $t$ with $\tau_{s-1}$ and
$v_{h_s}-\keep(h_s)$. If $s>1$ and $t<\tau_{s-1}$, then
\begin{equation}
 i=h_{s-1},
 \qquad z=\keep(h_{s-1})+t.
\label{eq:inverse-previous-heavy}
\end{equation}
If $\tau_{s-1}\le t<v_{h_s}-\keep(h_s)$, define
$\xi:=t+E(s-1)$. Then
$D(\rho(s-1))\le\xi<D(\rho(s))$, and the unique light rank and fragment
position are
\begin{equation}
 \begin{aligned}
 r&=\min\{u\in\{\rho(s-1)+1,\ldots,\rho(s)\}:D(u)>\xi\},\\
 z&=v_{\ell_r}+\xi-D(r-1).
 \end{aligned}
\label{eq:inverse-light}
\end{equation}
When these inequalities hold, the minimizing set in
\cref{eq:inverse-light} is nonempty. The inequalities above and the strict
increase of $D$ imply
$v_{\ell_r}\le z<Z$. Finally, if
$v_{h_s}-\keep(h_s)\le t<v_{h_s}$, then
\begin{equation}
 i=h_s,
 \qquad z=t-\bigl[v_{h_s}-\keep(h_s)\bigr].
\label{eq:inverse-own-heavy}
\end{equation}
For a light or balanced index $y$, \cref{eq:light-balanced-preimage} gives
$(i,z)=(y,t)$. The search in \cref{eq:inverse-light} uses the reversible
procedure for binary search described after \cref{eq:E-oracle}, applied to $D$.
For an explicit implementation, evaluate the alias data of $y$ to obtain
$\keep(y)$ and, when $s>1$, those of $h_{s-1}$ to obtain
$\tau_{s-1}=Z-\keep(h_{s-1})$. At most one additional search in $D$ is needed
for the middle case. It can search the full list of light ranks, since the
inequalities for $\xi$ already force its answer into the stated subinterval.
These are sequential searches, not an $n$-step search containing another
$n$-step search. With $C_{\rm search}$ as defined in
\cref{eq:search-cost-detail}, a conservative accounting is
\begin{equation}
 C_{R^{-1}}\le4C_{\rm gate}(\PIA)+2C_{\rm search}
             +O(d\Mcost(B)).
 \label{eq:R-inverse-cost-detail}
\end{equation}
The factors include reversing the intermediate evaluations after the answer
has been XORed into its target. The temporary registers are reused between searches.
This proves the bounds in both directions.
\end{proof}

\subsection{Implementing the fixed and variable uniform states}
\label{subsec:uniform-implementation}

For $m\ge1$, we prepare the uniform superposition over
$[m]=\{0,\ldots,m-1\}$ by setting the bits of the position register one at a
time~\cite{groverrudolph}. This construction can be represented by a binary
tree: a node specifies the bits already set. Before the final level, its
two children correspond to appending $0$ or $1$ as the next bit. For each choice, we count the
integers in $[m]$ whose binary representations begin with the resulting
sequence of bits. These two counts determine the amplitudes for the next bit.

At level $0\le k\le B$, write the first $k$ most significant bits of the
position register as $b_1\cdots b_k$, and define their integer value by
\[
 u:=\sum_{j=1}^{k}b_j2^{k-j},\qquad 0\le u<2^k.
\]
For $k=0$, the prefix is empty and $u=0$. The prefix $b_1\cdots b_k$
followed by $B-k$ zeros represents the integer $u2^{B-k}$. Allowing these
remaining bits to vary gives the set of consecutive integers
\begin{equation}
 \mathcal J_{k,u}=\{u2^{B-k},\ldots,(u+1)2^{B-k}-1\}.
 \label{eq:uniform-node-set}
\end{equation}
The number of elements of $\mathcal J_{k,u}$ belonging to $[m]$ is
\begin{equation}
 \begin{aligned}
 c_m(k,u)&=|\mathcal J_{k,u}\cap[m]|\\
 &=\min\{\max\{m-u2^{B-k},0\},2^{B-k}\}.
 \end{aligned}
 \label{eq:uniform-node-count}
\end{equation}
For $k<B$, appending a bit $\zeta\in\{0,1\}$ changes the integer value of
the prefix from $u$ to $2u+\zeta$. Thus
$a=c_m(k+1,2u)$ and $b=c_m(k+1,2u+1)$ count the integers in $[m]$ with
the current prefix followed by $0$ or $1$, respectively, and
$a+b=c_m(k,u)$. We use the following one-qubit gate:
\begin{equation}
 G(a,b)=
 \begin{cases}
 I,&b=0,\\
 X,&a=0<b,\\
 \Had,&a=b>0,\\
 R_y(2\theta),&a,b>0,\ a\ne b,
 \end{cases}
 \label{eq:node-unitary}
\end{equation}
where
\begin{equation}
 \theta=\arctan\sqrt{b/a},\qquad
 R_y(2\theta)=
 \begin{pmatrix}\cos\theta&-\sin\theta\\\sin\theta&\cos\theta\end{pmatrix}.
 \label{eq:node-angle-matrix}
\end{equation}
We use the convention $R_y(\phi)=e^{-i\phi Y/2}$, where $Y$ is the Pauli
$Y$ operator, so the required rotation angle is $2\theta$.
For nonempty nodes this gives
\begin{equation}
 G(a,b)\ket0=\sqrt{\frac a{a+b}}\ket0+\sqrt{\frac b{a+b}}\ket1.
 \label{eq:uniform-tree-rotation}
\end{equation}
Equivalently, on a nonempty node the angle is
\begin{equation}
 \vartheta_{m,k,u}=\arcsin\sqrt{\frac{c_m(k+1,2u+1)}{c_m(k,u)}}
 \quad(c_m(k,u)>0).
 \label{eq:uniform-tree-angle}
\end{equation}
For an empty node we use the identity. The case $a=0<b$ cannot arise for the
initial segment $[m]$, but makes the gate convention explicit.

At each level, apply the gate selected by $m$ and the prefix $u$ to the next
qubit, leaving the suffix untouched. Before that step, the amplitude of
prefix $u$ is $\sqrt{c_m(k,u)/m}$; multiplying by the amplitude assigned to
the next bit $\zeta$ gives $\sqrt{c_m(k+1,2u+\zeta)/m}$ for the extended
prefix. Induction over the levels
proves $U_m\ket0=m^{-1/2}\sum_{t<m}\ket t$.
For $m=0$ all nodes are empty and $U_0=I$.

This gate convention defines $U_m$ on \emph{every} basis state, including
states encountered when the circuit is inverted or approximated.
In particular, full nodes use $\Had$ in both the fixed and variable
implementations. Although $\Had\ket0=R_y(\pi/2)\ket0$, those two gates are
not the same unitary and must not be interchanged in a claim about operator norm.
The inverse applies the adjoint level gates in \emph{decreasing} order of
$k$. At each inverse step its input prefix is unchanged while the count
and angle are computed into ancilla registers and then uncomputed.

\subsubsection{Fixed capacity: implementing \texorpdfstring{$U_Z$}{UZ}}

The implementation is particularly simple for $m=Z$, because $Z$ is a known
classical constant. A set $\mathcal{J}_{k,u}$ disjoint from $[Z]$ uses the
identity, while a set contained in $[Z]$ is completed by Hadamard gates on its
remaining qubits. Only a partially filled set requires a non-Hadamard rotation,
and its count and angle are computed classically. The resulting fixed rotations
are synthesized before the quantum circuit is
executed; here synthesis means approximating each prescribed one-qubit
rotation by gates from a fixed finite universal gate set. Thus $U_Z$ uses only
fixed gates: it does not evaluate $S(N)$ or any binomial coefficient
coherently, and it does not store a list of $Z$ amplitudes. If $Z=2^q$ for a
nonnegative integer $q$, preparing $U_Z\ket{0^{B}}$ requires only
Hadamard gates on the $q$ least significant qubits, while the remaining
qubits stay at zero. This shortcut reproduces the action on the initial
state. Implementing the specified tree unitary on arbitrary inputs retains
the prefix controls described below.

More generally, let $\operatorname{nz}(Z)$ denote the number of nonzero bits in
the binary expansion of $Z$. This expansion decomposes
$\{0,\ldots,Z-1\}$ into $\operatorname{nz}(Z)$ disjoint sets of consecutive
integers whose sizes are powers of two. Controlled Hadamard gates prepare the
uniform state within each set, and $\operatorname{nz}(Z)-1$ rotations prepare
their relative amplitudes. Hence only these
$\operatorname{nz}(Z)-1\le B-1$ angles depend nontrivially on $Z$; the number
of remaining controlled bit flips and Hadamard operations is polynomial in
$B$.

An explicit conservative implementation makes the latter cost quadratic.
At depth $k$, write $Z=q_k2^{B-k}+r_k$, with
$0\le r_k<2^{B-k}$. A prefix $u<q_k$ is full and receives a
controlled Hadamard on the next qubit. The only possible partial prefix is
$u=q_k$ when $r_k>0$; it receives the fixed rotation determined by the counts
$a$ and $b$. Other prefixes use the identity. Testing a $k$-bit prefix
against these constants, using the flag, and uncomputing it costs $O(k+1)$
exact gates. Summing over levels gives $O(B^2)$, not a scan of
$Z$ integers. With target error $\delta$, a useful bound is
\begin{equation}
 C_Z(\delta)
 =O\!\left(B^2+
 B\,\mathsf S_{\rm rot}\!\left(\frac{\delta}{cB}\right)\right).
 \label{eq:UZ-explicit-cost}
\end{equation}
Only $\operatorname{nz}(Z)-1$ nontrivial rotations actually need synthesis.
For state preparation from zero, the case where $Z$ is a power of two uses exactly
$\log_2Z$ Hadamards, and $Z=1$ uses the identity. Leading zero bits can be
omitted from the register used for this preparation with fixed capacity.

Because $Z$ is classical, there is also an exact procedure using measurement.
Let $B_Z:=\lceil\log_2 Z\rceil$. Apply Hadamard gates to
$B_Z$ qubits, compute a
comparison bit that equals $1$ if and only if the resulting value $z$ satisfies
$z<Z$, and measure this bit. On outcome $1$, the register is exactly uniform on
$[Z]$ and is embedded in the $B$-qubit register holding the fragment position by
adjoining $B-B_Z$ leading zero qubits. On outcome $0$,
reinitialize the $B_Z$ qubits and repeat the procedure. The success
probability is $Z/2^{B_Z}>1/2$, so the
expected number of attempts is smaller than $2$. This alternative may be used
before the index register is involved, but it does not in general replace $U_{v_y}$ when
$y$ is in a superposition: a measurement whose success probability depends
on $y$ changes the output amplitudes. We use the unitary construction above in the deterministic preparation.

\subsubsection{Variable fibers: applying the inverse controlled by fiber volume}

By contrast, when $m=v_y$ and $y$ is in a superposition, the range is not a
classical constant. The direct sum in
\cref{eq:selected-uniform-inverse} does not mean that exponentially many
separate circuits are stored. One fixed $B$-level circuit computes
its controls from the registers holding $m$ and the current prefix.
At level $k$, the first $k$ qubits encode the integer $u$. The circuit computes the three
integer counts $c_m(k,u)$, $c_m(k+1,2u)$, and $c_m(k+1,2u+1)$ from the
registers holding $u$ and $m$. If $c_m(k,u)=0$, the controlled operation is
the identity. Otherwise, cases with $a=0$, $b=0$, or $a=b$
use the exact gates in \cref{eq:node-unitary}. For the remaining cases and a
target error $0<\epsilon<1$, we use reversible fixed-point
arithmetic, which retains a prescribed number of binary digits after the
binary point, to compute
\begin{equation}
 \vartheta_{m,k,u}=
 \arctan\sqrt{\frac{c_m(k+1,2u+1)}{c_m(k+1,2u)}}
\label{eq:uniform-tree-angle-interior}
\end{equation}
to additive error at most $\epsilon/(4B)$. Each binary digit of the computed
angle controls a fixed rotation about the $y$ axis whose angle is twice that
digit's place value. The combined synthesis error of these
rotations is also chosen to be at most $\epsilon/(4B)$. The counts and the
angle are then uncomputed. Appendix~\ref{app:coherent-angle-use} gives these
three steps explicitly, including the coherent computation of the angle
and the rotations controlled by its binary digits. For $U_{v_y}^\dagger$,
we reverse the order of the levels and apply the adjoints of the same
synthesized gates. Repeating these operations for the $B$ levels
uses the same ancilla registers at every level. All arithmetic used to
compute a count or angle is reversed while its input prefix and $m$ are
unchanged, so each arithmetic ancilla qubit returns to $\ket0$ even when the controlled rotation is
synthesized approximately. The synthesized rotations must leave the control
registers unchanged, so that reversing the arithmetic uses the same inputs
that computed the counts and angles. When a rotation's control condition is
false, its implementation must act exactly as the identity; this also keeps
empty branches unpopulated. Appendix~\ref{app:controlled-rotation-synthesis} constructs
controlled rotations with both properties. 

For clarity, a concrete conservative bound is stated here rather than hidden
in a generic polynomial. For a desired error $\delta$ in the circuit for
preparing a uniform state, put
\begin{equation}
 q=\left\lceil\log_2\frac{cB}{\delta}\right\rceil,
 \qquad L=O(B+q),
 \label{eq:uniform-precision-parameters}
\end{equation}
where $c>0$ is a sufficiently large absolute constant. Appendix~\ref{app:uniform-state-details}
gives an explicit fixed-point algorithm that uses range reduction to compute an angle to
$q$ fractional bits, using $O(qL^2)$ gates with schoolbook arithmetic.
At each level, there are $O(q)$ elementary rotations controlled by the
computed angle. Consequently,
\begin{equation}
 C_{\rm var}(\delta)
 =O\!\left(B q\left[
 L^2+\mathsf S_{\rm rot}\!\left(\frac{\delta}{cB q}\right)
 \right]\right).
 \label{eq:variable-uniform-explicit-cost}
\end{equation}
This is a deliberately non-optimal implementation bound. Faster angle
arithmetic is implemented in Appendix~\ref{app:fast-angle-computation}, and
\Cref{thm:fast-clean-preparation} gives the resulting complete cost. In all
cases, the quantum gate cost of coherently computing the angles from the
registers holding the fiber volume and the prefix must be included in the
total gate cost of clean state preparation with alias sampling.

\subsection{Clean polynomial state preparation}

Apply the reversible replacement procedure from the proof of
\Cref{prop:ranking-principle} with $(y,t)=R(i,z)$. First compute the output
index $y$ and fiber position $t$ in new registers. Then use $R^{-1}(y,t)$
to reconstruct $(i,z)$ and XOR it into the original registers holding the source index and
fragment position, returning both to zero. The ranking
oracles also return each arithmetic ancilla qubit to $\ket0$. The remaining
registers therefore contain
\begin{equation}
 \frac{1}{\sqrt{NZ}}
 \sum_{\substack{y\in[N]\\v_y>0}}\sum_{t=0}^{v_y-1}\ket{y}\ket{t}.
\label{eq:ordered-ranked-state}
\end{equation}
Compute $v_y=NP(y)^2$ from $y$ into a $B$-qubit register, apply
$\mathcal U^{\dagger}$ from \cref{eq:selected-uniform-inverse} with this
register controlling the inverse on the position register, and then
uncompute $v_y$. On the support of
\cref{eq:ordered-ranked-state}, the position register is
$U_{v_y}\ket{0^{B}}$ and is therefore returned to
$\ket{0^{B}}$. Omitting
registers that have returned to zero, the result is
\begin{equation}
 \sum_{\substack{y\in[N]\\v_y>0}}
 \sqrt{\frac{v_y}{NZ}}\ket{y}
 =\frac{1}{\sqrt Z}\sum_{y\in[N]}|P(y)|\ket{y}.
\label{eq:unsigned-clean-state}
\end{equation}
Finally, compute the bit that equals $1$ if and only if $P(y)<0$, apply a phase
of $-1$ when this bit is $1$, and uncompute it. This yields the state
$\ket{\psi}$ in \cref{eq:psi-P}.

\begin{theorem}[Clean preparation and its complete gate cost]
\label[theorem]{thm:clean-preparation}
Fix the integer polynomial $P$ and the preprocessed data. With ideal
rotations for preparing uniform states, the construction prepares
\begin{equation}
 \ket\psi=\frac{1}{\sqrt Z}\sum_{y\in[N]}P(y)\ket y
 \label{eq:clean-final}
\end{equation}
deterministically and returns every other register to zero.
The exact polynomial, alias, ranking, and sign computations cost
\begin{equation}
 C_{\rm exact}=O\bigl(nd\Mcost(B)\bigr).
 \label{eq:clean-arithmetic-cost}
\end{equation}
This term does not include either circuit for preparing a uniform state.

For $0<\epsilon<1$, there is a circuit over a finite gate set with joint output
$\ket{\widetilde\Psi}$ satisfying
\begin{equation}
 \left\|\ket{\widetilde\Psi}
       -\ket\psi\ket{0^{a_{\rm aux}}}_{\rm ancilla}\right\|_2\le\epsilon.
 \label{eq:clean-error}
\end{equation}
Here $a_{\rm aux}$ counts all qubits beyond the output register, including
the registers holding the source fragment, fiber position, and fiber
volume, and the ancilla qubits used for arithmetic. The circuit's total gate count is
\begin{equation}
 C_{\rm total}(\epsilon)
 =O\bigl(nd\Mcost(B)\bigr)
   +C_Z(\epsilon/2)+C_{\rm var}(\epsilon/2),
 \label{eq:clean-total-cost}
\end{equation}
where \cref{eq:UZ-explicit-cost,eq:variable-uniform-explicit-cost} give
explicit choices for the last two terms. In particular, the fully
schoolbook instantiation has the conservative bound
\begin{equation}
 \begin{split}
 C_{\rm total}(\epsilon)=O\!\bigg(&ndB^2+
 B qL^2\\
 &+B q\,\mathsf S_{\rm rot}
       \!\left(\frac{\epsilon}{cB q}\right)\bigg),
 \end{split}
 \label{eq:clean-schoolbook-cost}
\end{equation}
with $q=O(\log(B/\epsilon))$ and $L=O(B+q)$.
A conservative upper bound on the total number of qubits used beyond the
output register is $a_{\rm aux}=O(dB+n+L^2)$. No postselection is used.
\end{theorem}

\begin{proof}
The ideal state transformations preceding the theorem follow from
\Cref{prop:ranking-principle,prop:ordered-ranking}. To make the oracle
accounting explicit, after preparing the uniform superposition in
\cref{eq:uniform-fragments}, the implementation using the oracle convention
in \cref{eq:xor-oracle-convention} uses
\begin{equation}
 C_R+C_{R^{-1}}+2C_v+C_{\rm var}+C_{\rm sign}
 \label{eq:clean-call-ledger}
\end{equation}
gates. Here $C_v=O(d\Mcost(B))$ is the cost of evaluating
$v_y$, and $C_{\rm sign}=O(d\Mcost(B))$ includes computing,
applying, and uncomputing the sign phase. Add $n$ Hadamards and $C_Z$ for
the initial state. A separate coherent computation of $\pi$ is not required:
the forward ranking already computes $y=\pi(i,z)$. This gives
\cref{eq:clean-arithmetic-cost,eq:clean-total-cost} without double-counting
alias sampling.

Choose each implemented uniform circuit to differ from its specified ideal
unitary by at most $\epsilon/2$ in operator norm. A telescoping estimate for
the product of unitaries gives \cref{eq:clean-error}; all other operations
are exact. Within either uniform circuit, distribute its error among the
at most $B$ levels. The angle rounding and synthesized controlled
rotations at one level are each assigned part of that level's budget.
The circuit implementing $\mathcal U^\dagger$ preserves the $m$ register exactly, and each level
preserves its prefix controls. Its error is therefore the maximum of its
blockwise errors, not their sum over the possible values of $m$.
Thus there is no extra factor $N$ in the synthesis error bound.
The circuit for computing the angle and its bound on the number of ancilla
qubits are established in
Appendix~\ref{app:uniform-state-details}, yielding
\cref{eq:clean-schoolbook-cost}.
\end{proof}

\begin{theorem}[Main result: clean state preparation with fast arithmetic]
\label[theorem]{thm:fast-clean-preparation}
Fix the integer polynomial $P$ and the preprocessed data of
\Cref{thm:implicit-representation}. Let $0<\epsilon<1$, choose a sufficient
width $B$ as in \cref{eq:arithmetic-width}, and put
\begin{equation}
 q=\left\lceil\log_2\frac{cB}{\epsilon}\right\rceil,
 \qquad L=\Theta(B+q),
 \label{eq:fast-precision-parameters}
\end{equation}
with sufficiently large absolute constants. Here $q$ is the number of
fractional bits used to represent each rotation angle, and $L$ is a
sufficient width for the arithmetic computing it.
There is a deterministic circuit over Clifford+$T$ whose joint output
satisfies
\begin{equation}
 \left\|\ket{\widetilde\Psi}
       -\ket\psi\ket{0^{a_{\rm aux}}}_{\rm ancilla}\right\|_2\le\epsilon,
 \label{eq:clean-fast-error}
\end{equation}
where $\ket\psi$ is the state in \cref{eq:psi-P} and
$a_{\rm aux}=O(dB+n+qL)$ is the number of qubits beyond the output register.
Using the fast multiplication circuit in
\cref{eq:quasilinear-multiplication}, its complete gate count is
\begin{equation}
 \begin{split}
 C_{\rm fast}(\epsilon)=O\!\bigg(&ndB\log(2B)
       +BqL\log(2L)\\
       &+Bq\,\mathsf S_{\rm rot}
          \!\left(\frac{\epsilon}{cBq}\right)\bigg).
 \end{split}
 \label{eq:clean-fast-cost}
\end{equation}
The three terms account for exact polynomial and alias arithmetic, angle computation, and
rotation synthesis, respectively. Classical preprocessing is performed
once for the alias oracle and has the polynomial cost stated in
\Cref{thm:implicit-representation}.

In particular, if $\epsilon\ge B^{-a}$ for a fixed constant $a>0$, then
$q=O(\log B)$ and $L=\Theta(B)$. With
$\mathsf S_{\rm rot}(\eta)=\operatorname{poly}(\log(1/\eta))$, the complete
gate count becomes
\begin{equation}
 C_{\rm fast}(\epsilon)=\Ot(ndB+B^2),
 \label{eq:clean-fast-soft-cost}
\end{equation}
where $\Ot$ suppresses factors polynomial in $\log B$.
\end{theorem}

\begin{proof}
Appendix~\ref{app:fast-angle-computation} implements the rounded angle
calculation using fast multiplication, division, and square roots. One
angle uses $O(q\Mcost(L))=O(qL\log(2L))$ gates and $O(qL)$ ancilla
qubits. The $B$ levels therefore contribute $O(BqL\log(2L))$ gates for
angle computation and $O(Bq)$ synthesized controlled rotations. Giving
each such rotation error at most $\epsilon/(cBq)$ bounds their total error
by the assigned part of the error budget. The arithmetic registers are
reused between levels.

The exact computations contribute $O(ndB\log(2B))$ gates by
\cref{eq:clean-arithmetic-cost,eq:quasilinear-multiplication}.
For $U_Z$, the $O(B^2)$ exact gates in \cref{eq:UZ-explicit-cost} are
absorbed by $O(BqL\log(2L))$. Its $O(B)$ nontrivial rotations can use the
same synthesis tolerance as the variable circuit, so their cost is also
absorbed by the last term of \cref{eq:clean-fast-cost}. Apply
\Cref{thm:clean-preparation}, assigning error $\epsilon/2$ to each uniform
circuit and adjusting $c$, to obtain \cref{eq:clean-fast-error}.

The exact arithmetic uses $O(dB+n+\mathsf W(B))$ additional qubits, and
the angle circuit uses $O(qL)$. Since $\log(2L)=O(q)$ and $B=O(L)$,
$\mathsf W(B)=O(B\log(2B))$ and the remaining $O(B+n)$ qubits fit within
$a_{\rm aux}$. Finally, the stated error regime makes every synthesis
cost polynomial in $\log B$, yielding \cref{eq:clean-fast-soft-cost}.
\end{proof}

\begin{corollary}[Complete gate counts in two precision regimes]
\label[corollary]{cor:precision-regimes}
Use the implementation of \Cref{thm:fast-clean-preparation}, with
$\mathsf S_{\rm rot}(\eta)=O(\log(2/\eta))$ and the sufficient width $B$
in \cref{eq:explicit-arithmetic-width}. Its complete gate count is
\begin{equation}
 C_{\rm fast}(\epsilon)
 =O\!\left(ndB\log(2B)
   +Bq(B+q)\log\bigl(2(B+q)\bigr)\right),
 \label{eq:complete-logarithmic-synthesis-cost}
\end{equation}
where $q$ is defined in \cref{eq:fast-precision-parameters}.
Hold $n,d$ fixed as $\epsilon\to0$.

For fixed coefficient data, $B$ is fixed and
\begin{equation}
 \begin{split}
 C_{\rm fast}(\epsilon)
 &=O\!\bigl(ndB\log(2B)\\
 &\qquad+B\log^2(1/\epsilon)\log\log(1/\epsilon)\bigr)\\
 &=O(B\log^2(1/\epsilon)\log\log(1/\epsilon)).
 \end{split}
 \label{eq:fixed-coefficient-precision-cost}
\end{equation}
In the regime $\beta+d\log(d+1)=O(nd)$, this gives
\begin{equation}
 C_{\rm fast}(\epsilon)=O(nd\,\log^2(1/\epsilon)\log\log(1/\epsilon)).
 \label{eq:fixed-coefficient-nd-cost}
\end{equation}

For coefficient data refined with $\beta=\Theta(\log(1/\epsilon))$,
the chosen width satisfies $B=\Theta(\log(1/\epsilon))$ and
\begin{equation}
 \begin{split}
 C_{\rm fast}(\epsilon)
 &=O\!\bigl(nd\,\log(1/\epsilon)\log\log(1/\epsilon)\\
 &\qquad+\log^3(1/\epsilon)\log\log(1/\epsilon)\bigr)\\
 &=O(\log^3(1/\epsilon)\log\log(1/\epsilon)).
 \end{split}
 \label{eq:refined-coefficient-precision-cost}
\end{equation}
Both bounds include coherent quantum computation of rotation angles,
rotation synthesis, and the inverse operations used for uncomputation.
\end{corollary}

\begin{proof}
Appendix~\ref{app:controlled-rotation-synthesis} establishes the logarithmic
synthesis bound. Since $q=\lceil\log_2(cB/\epsilon)\rceil$, the synthesis
term in \cref{eq:clean-fast-cost} satisfies
\begin{equation}
 Bq\,\mathsf S_{\rm rot}\!\left(\frac{\epsilon}{cBq}\right)
 =O\!\left(Bq\log\frac{2cBq}{\epsilon}\right)
 =O(Bq^2).
 \label{eq:logarithmic-synthesis-contribution}
\end{equation}
With $L=\Theta(B+q)$, this term is absorbed by the coherent angle
computation, giving \cref{eq:complete-logarithmic-synthesis-cost}.
For fixed $B$, we have $q=\Theta(\log(1/\epsilon))$ and
$L=\Theta(\log(1/\epsilon))$ as $\epsilon\to0$, which gives
\cref{eq:fixed-coefficient-precision-cost}. The stated condition on
$\beta,d,n$ makes the width in \cref{eq:explicit-arithmetic-width}
$O(nd)$, proving \cref{eq:fixed-coefficient-nd-cost}.
If $\beta=\Theta(\log(1/\epsilon))$, that width is $\Theta(\log(1/\epsilon))$;
then $q,L=\Theta(\log(1/\epsilon))$ and
\cref{eq:refined-coefficient-precision-cost} follows. In this second
regime, rotation synthesis contributes $O(\log^3(1/\epsilon))$ gates and
coherent angle computation contributes $O(\log^3(1/\epsilon)\log\log(1/\epsilon))$.
\end{proof}

For a fixed real polynomial $f$ and fixed grid with a nonzero sampled
vector, the coefficient rounding in Appendix~\ref{app:input-accuracy}
connects the second regime to the total preparation error. Assign
$\epsilon_{\mathrm{coef}}=\epsilon/2$ in
\cref{eq:coefficient-approximation-width}. Since $\chi_N(f)$ is fixed,
this gives $b=\log_2(1/\epsilon)+O(1)$. The normalized coefficient of largest
absolute value rounds to $\pm2^b$, so this choice has
$\beta=b+1=\Theta(\log(1/\epsilon))$. Applying the circuit with error tolerance
$\epsilon/2$ gives total joint error at most $\epsilon$ by
\cref{eq:dyadic-state-error,eq:clean-fast-error}, with the complete gate
count $O(\log^3(1/\epsilon)\log\log(1/\epsilon))$ in
\cref{eq:refined-coefficient-precision-cost}.

\begin{remark}[What ``garbage-free'' means at finite precision]
The ideal circuit returns the full ancilla register to $\ket{0^{a_{\rm aux}}}$.
For synthesized rotations, \cref{eq:clean-error} is a statement about the
\emph{joint} state. Each arithmetic ancilla qubit returns exactly to
$\ket0$, while the register holding the fiber position is only close to
$\ket{0^B}$ and can retain a small entanglement with $y$. The reduced
output state $\widetilde\rho$ satisfies
\begin{equation}
 \tfrac12\bigl\|\widetilde\rho-\ket\psi\!\bra\psi\bigr\|_1
 \le\epsilon.
 \label{eq:reduced-clean-error}
\end{equation}
If the integer polynomial approximates a real polynomial with error
$\epsilon_{\rm coef}$ as in \cref{eq:dyadic-state-error}, the joint error
relative to the target defined by that real polynomial is at most
$\epsilon_{\rm coef}+\epsilon$. These are separate accuracy budgets.
\end{remark}

The compressed representation is exponentially smaller than an
explicit $N$-entry table when the input parameters are polynomially bounded
in $n$; that comparison alone is not a quantum--classical sampling separation.

\begin{remark}[Scope of the comparison]
The resource bounds establish an efficient, compressed \emph{alias}
implementation and its reversible cleanup; they are not lower bounds for
preparing polynomial states. The same exact cumulative function $S$ can
also support direct cumulative inversion or direct rotations determined by counts.
Thus an advantage over those alternatives requires a separate resource or
application comparison. The present construction isolates what is needed
to make an ordered alias sampler itself garbage-free.
\end{remark}

\section{Conclusion}
\label{sec:conclusion}
In the paper, we exponentially reduced the computational complexity of alias sampling algorithm with a structured quantum state—the quantum state with amplitude generated by a finite-degree polynomial. Both the classical preprocessing and coherent quantum complexity have been exponentially reduced from $\mathcal{O}(N)$ to $poly(n,d)$. Such reduction is made possible by a more efficient representation of alias tables for polynomial states. Instead of loading all amplitudes of the polynomial state, we can merely use the low-dimensional description of the polynomial and coherently compute the alias table on the quantum computer. In the high-limit regime $\epsilon \to 0$, the gate complexity is of $\tilde{\mathcal{O}}(\log^3(1/\epsilon))$, independent of $n,d$. In addition, when $\beta$ is a fixed number, the leading order of complexity matches the state-of-the-art method in terms of leading order, which is linear in $n,d$. 

The compact representation exploits two features of polynomial amplitudes.
First, the degree bound limits the light, heavy, and balanced indices to
$O(d)$ groups of consecutive indices, so their boundaries and types can be
stored in a small set of group records. Second, polynomial values and their
cumulative squared weights admit short exact formulas. In particular,
$S(t)=\sum_{i<t}P(i)^2$ is a polynomial of degree at most $2d+1$, with
integer coefficients in the binomial basis. These formulas let us evaluate
the total weight over an interval without summing its entries individually.
Together, the group records and exact cumulative formulas give the deficits
and surpluses used by reversible searches to compute each alias and keep
threshold. We can therefore evaluate the alias table coherently from the
polynomial coefficients without constructing or storing its $2^n$ entries.

Moreover, for the first time, it is possible to prepare quantum states with alias sampling without garbage states. Such clean state preparation techniques can be useful beyond the scope of polynomial state preparation and commonly use places like quantum input models and block encoding. This gives a new path for the initial state preparation with sampling type methods, which is extremely helpful for quantum linear solver, quantum differential equation solver.

There are still a few directions left open along the line of polynomial state preparation or alias sampling. The current method requires intensive arithmetic computation. There can be a large constant overhead for naive implementations, and further reducing the complexity can be of great interest. Alias sampling for other structured states, including sparse states without garbage states, can also be interesting. Since the method can also be considered a quantum-inspired classical algorithm, it would be intriguing to compare it with other classical sampling algorithms, particularly those based on diffusion models or stochastic differential equations.

\section*{Acknowledgments}
The authors thank Nathan Baker, Daan Camps, Brad Lackey, Burak Sahinoglu, Matthias Christandl, and Matthias Troyer for valuable discussions.
D.L and WA.dJ was supported by the U.S. Department of Energy, Office of Science, Office of Advanced Scientific Computing Research through the Accelerated Research in Quantum Computing Program MACH-Q Project. H.W. and J.C. acknowledge support from the NSF NQVL:QSTD Phase II
program under Award No.~2533041. C.Y was supported by the U.S. Department of Energy, Office of Science, Office of Advanced Scientific Computing Research through Quantum Utility through Advanced Computational Quantum Algorithms and the Quantum System Accelerator projects.


The authors acknowledge the use of GPT-6 Astra for assistance with literature searches, language editing and preparation of initial note.

\appendix

\section{Worked example of the ordered table}
\label{app:worked-example}

Let $N=8$ and $P(i)=i$ for every $i\in[N]$. Then
\begin{equation}
 Q(i)=i^2,
 \qquad
 Z=\sum_{i=0}^{7}i^2=140,
 \qquad
 v_i=8i^2.
\label{eq:example-data}
\end{equation}
The light, heavy, and balanced index sets in \cref{eq:LHB} are
\begin{equation}
 \mathcal{I}_{\mathcal{L}}=\{0,1,2,3,4\},
 \qquad
 \mathcal{I}_{\mathcal{H}}=\{5,6,7\},
 \qquad
 \mathcal{I}_{\mathcal{B}}=\varnothing.
\label{eq:example-index-sets}
\end{equation}
Thus
\begin{equation}
 (\ell_1,\ldots,\ell_5)=(0,1,2,3,4),
 \qquad
 (h_1,h_2,h_3)=(5,6,7).
\label{eq:example-ordered-indices}
\end{equation}
The cumulative missing and extra weights are
\begin{equation}
 \begin{aligned}
 (D(0),\ldots,D(5))&=(0,140,272,380,448,460),\\
 (E(0),\ldots,E(3))&=(0,60,208,460).
 \end{aligned}
\label{eq:example-DE}
\end{equation}
The coefficients obtained from finite differences of $Q$ are
$(\kappa_0,\kappa_1,\kappa_2)=(0,1,2)$, and therefore
\begin{equation}
 S(t)=\binom{t}{2}+2\binom{t}{3}.
\label{eq:example-S}
\end{equation}
Substituting this prefix sum into \cref{eq:D-closed,eq:E-closed} produces
every value in \cref{eq:example-DE} without summing over the eight indices.

The definitions in \cref{eq:s-of-r,eq:rho-of-s,eq:ordered-table} give
\begin{equation}
\begin{array}{c|r|c|r|c}
i&v_i&\text{type}&\keep(i)&\alias(i)\\
\hline
0&0   &\mathcal{L}&0   &5\\
1&8   &\mathcal{L}&8   &6\\
2&32  &\mathcal{L}&32  &7\\
3&72  &\mathcal{L}&72  &7\\
4&128 &\mathcal{L}&128 &7\\
5&200 &\mathcal{H}&60  &6\\
6&288 &\mathcal{H}&76  &7\\
7&392 &\mathcal{H}&140 &7
\end{array}
\label{eq:example-table}
\end{equation}
For example, the heavy index $6$ receives $\keep(6)=76$ retained fragments,
$140-8=132$ fragments transferred from the light index $1$, and
$140-60=80$ fragments transferred from the preceding heavy index $5$. Their
sum is $288=v_6$, as required by \cref{eq:balance-condition}.

For this example,
\begin{equation}
 \begin{aligned}
 (s(1),\ldots,s(5))&=(1,2,3,3,3),\\
 (\rho(0),\rho(1),\rho(2),\rho(3))&=(0,1,2,5),\\
 (\tau_0,\tau_1,\tau_2,\tau_3)&=(0,80,64,0).
 \end{aligned}
\label{eq:example-ranks-and-tau}
\end{equation}
For each $s\in\{1,2,3\}$, the fragments transferred from the preceding heavy
index when $s>1$, the fragments transferred from light indices, and the
fragments retained by the heavy index $h_s$ occupy the following ranges of the
position $t$:
\begin{equation}
\begin{array}{c|c|c|c|c}
h_s&\substack{\text{from preceding}\\\text{heavy index}}
&\text{from light indices}&\text{retained by }h_s&v_{h_s}\\
\hline
5&\varnothing&0\le t<140&140\le t<200&200\\
6&0\le t<80&80\le t<212&212\le t<288&288\\
7&0\le t<64&64\le t<252&252\le t<392&392
\end{array}
\label{eq:example-rank-ranges}
\end{equation}
For the heavy index $6$, the source fragments $(i,z)=(5,100)$, $(1,50)$,
and $(6,30)$ have fiber positions $t=40$, $122$, and $242$, respectively.
They illustrate a
fragment transferred from the preceding heavy index, one transferred from a
light index, and one retained by index $6$. The output index $0$ has an empty
fiber, so no term with $y=0$ occurs in the ranked state.

To see the inverse ranking explicitly, take $(y,t)=(6,122)$. The block
transferred from the preceding heavy index has length $\tau_1=80$, and the
retained block begins at $v_6-\keep(6)=212$. Since $80\le122<212$, the
fragment comes from a light index. Then
\begin{equation}
 \begin{gathered}
 \xi=122+E(1)=182,\\
 D(1)=140\le182<D(2)=272.
 \end{gathered}
\end{equation}
Thus $r=2$, $i=\ell_2=1$, and
$z=v_1+182-D(1)=8+42=50$. Hence
$R^{-1}(6,122)=(1,50)$, as required.
After ranking, conditioned on output $y=6$, the register holding the fiber
position is in the state
\begin{equation}
 \frac1{\sqrt{288}}\sum_{t=0}^{287}\ket t
 =U_{288}\ket0,
\end{equation}
Applying $U_{288}^{\dagger}$ resets this register to $\ket0$. The output
amplitude remains
$\sqrt{288/(8\cdot140)}=6/\sqrt{140}$.

\section{Additional implementation details}
\label{app:implicit-oracle-details}

\subsection{Selection from the stored group records}
\label{subsec:group-selection}

We use the group records of \Cref{thm:implicit-representation}.
The boundaries $\gamma_{j-1}$ and $\gamma_j$ are defined in
\cref{eq:group-boundaries}. The counts $\lambda_j$ and $\eta_j$ are the
numbers of light and heavy indices before $\gamma_{j-1}$, respectively,
as defined in \cref{eq:group-earlier-counts}. The corresponding cumulative
missing and extra weights $D(\lambda_j)$ and $E(\eta_j)$ are obtained
from \cref{eq:DE}.

For an input $i\in[N]$ and each $a\in\{1,\ldots,G-1\}$, compute a bit
$\chi_a$ that equals $1$ if and only if
$i\ge\gamma_a$. Since the boundaries are increasing,
\begin{equation}
 j(i):=1+\sum_{a=1}^{G-1}\chi_a
\label{eq:appendix-group-number}
\end{equation}
is the unique group number satisfying
$\gamma_{j(i)-1}\le i<\gamma_{j(i)}$.

For $j\in\{1,\ldots,G\}$, let $u_j$ be any integer value stored with group
$j$. Then
\begin{equation}
 u_{j(i)}=u_1+\sum_{a=1}^{G-1}
 \chi_a(u_{a+1}-u_a).
\label{eq:appendix-group-selection}
\end{equation}
The first $j(i)-1$ comparison bits equal $1$ and the remaining bits equal
$0$, so the differences in the sum add up to $u_{j(i)}-u_1$. Adding $u_1$
therefore gives $u_{j(i)}$. Applying
\cref{eq:appendix-group-selection} selects
$\gamma_{j(i)-1}$, $\gamma_{j(i)}$, $\lambda_{j(i)}$, $\eta_{j(i)}$,
$D(\lambda_{j(i)})$, and $E(\eta_{j(i)})$. We encode the three possible types
by the two-bit strings $\mathcal{L}\leftrightarrow(0,0)$,
$\mathcal{H}\leftrightarrow(0,1)$, and $\mathcal{B}\leftrightarrow(1,0)$ and apply
\cref{eq:appendix-group-selection} to each encoding bit to select
$\theta_{j(i)}$. Thus all data associated with the group are selected using
$O(G)=O(d)$ comparisons and controlled additions.

For $j\in\{1,\ldots,G\}$ and $i\in[N]$ satisfying
$\gamma_{j-1}\le i<\gamma_j$, \cref{eq:light-rank} gives the light rank when
$\theta_j=\mathcal{L}$. When $\theta_j=\mathcal{H}$, the heavy rank is
\begin{equation}
 \operatorname{rank}_{\mathcal{I}_{\mathcal{H}}}(i)
 =\eta_j+i-\gamma_{j-1}+1.
\label{eq:appendix-heavy-rank}
\end{equation}
Conversely, for $r\in\{1,\ldots,m_{\mathcal{L}}\}$, let
$j_{\mathcal{L}}(r)$ identify the unique light group specified below. For
$s\in\{1,\ldots,m_{\mathcal{H}}\}$, define $j_{\mathcal{H}}(s)$ analogously for
the heavy groups:
\begin{equation}
 \begin{aligned}
  \theta_{j_{\mathcal{L}}(r)}&=\mathcal{L},&
  \lambda_{j_{\mathcal{L}}(r)}&<r
  \le\lambda_{j_{\mathcal{L}}(r)}
  +\gamma_{j_{\mathcal{L}}(r)}-\gamma_{j_{\mathcal{L}}(r)-1},\\
  \theta_{j_{\mathcal{H}}(s)}&=\mathcal{H},&
  \eta_{j_{\mathcal{H}}(s)}&<s
  \le\eta_{j_{\mathcal{H}}(s)}
  +\gamma_{j_{\mathcal{H}}(s)}-\gamma_{j_{\mathcal{H}}(s)-1}.
 \end{aligned}
\label{eq:appendix-inverse-rank-groups}
\end{equation}
After the appropriate group is selected, one addition recovers the index. All
constants used here are classical values inserted when the circuit is
constructed. More explicitly, the groups identified in
\cref{eq:appendix-inverse-rank-groups} give
\begin{equation}
 \begin{aligned}
  \ell_r&=\gamma_{j_{\mathcal{L}}(r)-1}+r
  -\lambda_{j_{\mathcal{L}}(r)}-1,\\
  h_s&=\gamma_{j_{\mathcal{H}}(s)-1}+s
  -\eta_{j_{\mathcal{H}}(s)}-1.
 \end{aligned}
\label{eq:appendix-rank-to-index}
\end{equation}

\subsection{From isolated roots to grid groups}
\label{app:root-to-grid}

The goal of this classical preprocessing step is to compute the group
boundaries $\gamma_j$ in \cref{eq:group-boundaries} directly from the
polynomial coefficients. Recall from \cref{eq:classification-polynomial}
that $F(i)=NP(i)^2-Z$ determines whether an integer index $i\in[N]$ is
light, heavy, or balanced according to whether its value is negative,
positive, or zero. The sign of a real polynomial is constant between
consecutive real roots. We therefore locate the roots relative to the
integer grid and use them to identify the groups of consecutive indices
with the same classification. The resulting group records become fixed
constants in the quantum circuit. If $F$ is constant, its sign determines
a single group; in the following, assume $r:=\deg F\ge1$.

We first compute the square-free part $\widehat F$ of $F$, obtained by
removing repeated copies of its factors. This polynomial has the same
distinct roots as $F$, and each of its roots is simple. A certified classical
root-isolation algorithm~\cite{sagraloff2011} produces disjoint closed
intervals with rational endpoints, each containing exactly one distinct real root
and together containing all real roots. Each such interval is called an
\emph{isolating interval}. We choose endpoints that are not roots and refine
the intervals to width less than $1/2$. Intervals disjoint from $[0,N-1]$
can then be discarded. To refine a remaining interval, evaluate
$\widehat F$ exactly at its midpoint and retain the half whose endpoints
have opposite signs. If the midpoint is itself a root, its position is
known exactly. This bisection procedure preserves the unique root in the
interval. The signs used for the final light, heavy, and balanced
classification are evaluated using the original polynomial $F$.

The next step determines where each isolated root lies relative to the
integers. If its interval lies entirely between two consecutive integers
$k$ and $k+1$, the root can contribute a group boundary at $k+1$.
Otherwise, the interval contains a unique candidate integer $j$, because
its width is less than $1/2$. We evaluate $F(j)$ exactly. If $F(j)=0$,
then $j$ is the unique root in the interval and is a balanced index whenever
$j\in[N]$. If $F(j)\ne0$, we continue bisection until the interval lies
entirely on one side of $j$. We now bound the precision required for this
last step.

Let $H_F\ge1$ bound the magnitudes of the integer coefficients of $F$.
Every retained interval lies within $[-1,N+1]$. On this range,
\begin{equation}
 |F'(x)|\le H_F\frac{r(r+1)}2(N+2)^{r-1}=:K_F.
 \label{eq:root-grid-derivative}
\end{equation}
If $F(j)\ne0$, its integer value satisfies $|F(j)|\ge1$. For the root
$\alpha$ in the interval, the mean value theorem therefore gives
\begin{equation}
 1\le |F(j)-F(\alpha)|\le K_F|j-\alpha|,
 \qquad |\alpha-j|\ge K_F^{-1}.
 \label{eq:root-integer-separation}
\end{equation}
An interval of width less than $K_F^{-1}$ containing $\alpha$ therefore
excludes every integer that is not a root. It suffices to refine to width
at most $2^{-L}$, where
\begin{equation}
 \begin{aligned}
 L&:=\lceil\log_2 K_F\rceil+1,\\
 L&=O\bigl(\log H_F+rn+\log(r+1)\bigr).
 \end{aligned}
 \label{eq:root-grid-precision}
\end{equation}
Here we used $N=2^n$. The magnitude of $K_F$ may grow exponentially with
$n$, while the required number of binary digits is proportional to
$\log_2 K_F$. Starting from an interval of width $W_0<1/2$, after $t$
bisections its width is $W_0 2^{-t}$. Thus $O(L)$ bisections suffice to
reach the required width.

The cost of the arithmetic in each bisection is also polynomial.
For the integer input model in \cref{eq:integer-P,eq:coefficient-width},
one may take
\begin{equation}
 H_F=(d+1)^2 2^{2\beta}N^{2d+1}.
 \label{eq:root-coefficient-height}
\end{equation}
Indeed, this bounds each coefficient of $NP^2$ and the normalization $Z$;
the constant coefficient of $F$ is the difference of two nonnegative
quantities bounded by $H_F$. Consequently, $r\le2d$ and
$\log H_F=O(\beta+dn+\log(d+1))$, so $L$ is polynomial in $d$, $n$, and
$\beta$. Computing the square-free part and isolating its roots have
polynomial bit complexity~\cite{sagraloff2011}. The rational endpoints
produced by this computation have polynomial bit lengths, and each
bisection adds at most one bit to their common denominator. Exact
polynomial evaluation at these rational points also uses integers with
polynomial bit lengths. Hence both the number of refinement steps and
the bit cost of each step are polynomial in the input parameters.

Finally, we construct the groups from the root locations. Include $0$ and
$N$ as boundary candidates. A noninteger root between $k$ and $k+1$ adds
$k+1$ as a candidate, and an integer root $j\in[N]$ adds $j$ and $j+1$
to isolate the balanced index. Retain candidates in $[0,N]$, sort them,
and remove duplicates. Evaluate the original $F$ at one integer in each
resulting nonempty group, and merge neighboring groups with equal signs.
This also handles repeated roots at which the sign stays the same and
multiple roots between the same pair of integers. There are $O(d)$ roots
and boundary candidates, so these evaluations determine all classification
groups with polynomial classical cost and avoid a scan of all $N=2^n$
indices.

\subsection{Exact division by a known constant}
\label{subsec:constant-division}

This subsection constructs the exact division required by the binomial
recurrence in \cref{eq:binomial-recurrence}. The divisor $h\ge1$ is a
classical constant fixed when the circuit is constructed; in the recurrence,
$h=k+1$. The numerator $w$ is held in an $L$-qubit register, with $L\ge1$,
$0\le w<2^L$, and the promise that $h$ divides $w$. For an arbitrary
$L$-bit output string $y$, the desired transformation is
\begin{equation}
 \begin{aligned}
 &\ket{w}\ket{y}\ket{0^{a_{\rm div}}}\\
 &\quad\longmapsto
 \ket{w}\ket{y\mathbin{\oplus}(w/h)}\ket{0^{a_{\rm div}}},
 \end{aligned}
 \label{eq:constant-division-map}
\end{equation}
where $\oplus$ is bitwise XOR and $a_{\rm div}$ counts all ancilla qubits.
The input is preserved, and every ancilla qubit returns to $\ket0$.
If $h\ge2^L$, the divisibility promise forces $w=0$, so the identity circuit
implements this map. We therefore consider $1\le h<2^L$ below.

Write $h=2^\nu o$, where $o$ is odd and $0\le\nu<L$, and let $q=w/h$.
Divisibility gives
\[
 w=2^\nu o q,\qquad a:=w/2^\nu=oq.
\]
Thus the low $\nu$ bits of $w$ are zero, and its high $L-\nu$ bits encode
$a$. Since $o$ is odd, it has an inverse modulo $2^L$. Compute the constant
$c\in\{0,\ldots,2^L-1\}$ satisfying $oc\equiv1\pmod{2^L}$ classically
in advance. Then $ac=oqc\equiv q\pmod{2^L}$. Moreover,
$0\le q\le w<2^L$, so this residue equals the desired integer quotient:
\begin{equation}
 \frac wh=\left[\left(\frac{w}{2^\nu}\right)c\right]\bmod 2^L.
 \label{eq:constant-division-modular}
\end{equation}
The division can therefore be implemented by copying bits and multiplying
by a fixed constant.

We use the clean multiplication circuit specified in the arithmetic
conventions, which preserves its input registers and writes their full
product into a separate register. The following construction makes all
registers explicit:
\begin{enumerate}
 \item \emph{Form the shifted operand.}
 Initialize an $L$-qubit register $A$ to $\ket{0^L}$. Using $L-\nu$
 CNOT gates, copy the high $L-\nu$ bits of $w$ into the low $L-\nu$ bits
 of $A$. Its high $\nu$ bits remain zero, so $A$ holds $a=w/2^\nu$.

 \item \emph{Compute the product.}
 Initialize an $L$-qubit register $C$ to the known constant $c$ using
 at most $L$ NOT gates. Apply clean multiplication to $A$ and $C$, writing
 the full product $ac$ into a $2L$-qubit register $R$ initially in
 $\ket{0^{2L}}$. By \cref{eq:constant-division-modular}, the low $L$ bits
 of $R$ encode $q$.

 \item \emph{Copy the quotient to the output.}
 XOR the low $L$ bits of $R$ into the register holding $y$, using $L$
 CNOT gates. That register now holds $y\oplus q$.

 \item \emph{Uncompute.}
 Apply the inverse multiplication circuit to return $R$ to
 $\ket{0^{2L}}$. Undo the NOT gates that prepared $C$ and the CNOT gates
 that prepared $A$, returning both registers to $\ket{0^L}$.
 The multiplication circuit also returns each of its internal ancilla
 qubits to $\ket0$.
\end{enumerate}
All these operations act coherently when $w$ is in superposition. The
same reversible arithmetic gives a deterministic extension to inputs
outside the divisibility promise, as in \cref{eq:xor-oracle-convention}.

The construction uses one multiplication and its inverse, together with
$O(L)$ NOT and CNOT gates. Since $\Mcost(L)\ge L$, the total cost is
$O(\Mcost(L))$ gates. The registers $A$, $C$, and $R$ use $4L$ ancilla
qubits in total, in addition to the multiplication circuit's internal
ancilla qubits. Thus
\begin{equation}
 a_{\rm div}\le4L+\mathsf W(L)=O\bigl(L+\mathsf W(L)\bigr).
 \label{eq:constant-division-space}
\end{equation}
Here the register holding $w$ and the output register holding $y$ are
excluded from $a_{\rm div}$. Keeping the full product until uncomputation
makes its contribution to the qubit count explicit.

For the binomial recurrence, the numerator is
$w=\binom{t}{k}(t-k)$ on the active branch $k<t$, and $h=k+1$.
The width in \cref{eq:explicit-arithmetic-width} bounds this full numerator
before division and supplies a valid choice of $L$. Inactive branches are
assigned zero directly. Each division therefore costs $O(\Mcost(L))$
gates and uses $O(L+\mathsf W(L))$ ancilla qubits, which can be reused
at subsequent steps of the recurrence.


\section{Gate bounds for uniform states with variable range}
\label{app:uniform-state-details}

Clean state preparation uses $U_Z$ to prepare the sampling input and
a uniform-state inverse to erase the fiber position.
For a $B$-qubit position register and $1\le m<2^B$, the ideal circuit
$U_m$ prepares
\[
 U_m\ket{0^B}=\frac{1}{\sqrt m}\sum_{t=0}^{m-1}\ket t.
\]
The sampling circuit uses the classically known range $m=Z$.
For fiber erasure, the range is $m=v_y$, the number of source fragments
assigned to output $y$. This value is supplied in a quantum register
that controls $U_{v_y}^\dagger$ and may be in superposition.
We set $U_0=I$ to include empty ranges.

Both operations use the binary decision tree in
\cref{fig:uniform-state}(a), whose branches set successive position bits.
Its node gates can be precomputed for fixed $Z$, whereas a variable
$v_y$ requires their parameters to be computed coherently from the range
and the preceding bits. Each computed angle must then be uncomputed
after use, as in \cref{fig:uniform-state}(b). The gate bounds therefore
include this reversible arithmetic as well as rotation synthesis.

\begin{figure*}[t]
\centering
\includegraphics[width=\textwidth]{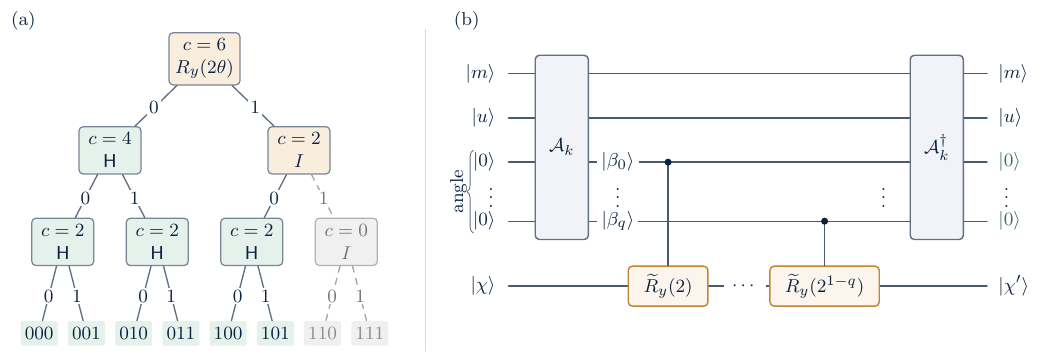}
\caption{Uniform-state preparation: (a) binary tree for $m=6$, with counts
$c$ and node gates, and (b) coherent rotations with exact angle uncomputation.}
\label{fig:uniform-state}
\end{figure*}

\subsection{Error allocation and the circuit for fixed capacity}

We divide the total tolerance $\epsilon$ equally between $U_Z$ and
$\mathcal U^\dagger$, which applies $U_{v_y}^\dagger$ under the control
of the volume register. For either circuit, let $0<\delta<1$ denote its
operator-norm error budget, including relative phase errors between
output labels. Within a level, the node gates act on disjoint control
subspaces, so their errors contribute through the largest node error.
Across the at most $B$ levels, the errors add under composition.
A sufficiently small constant times $\delta/B$ per node therefore
suffices for the circuit budget.

The binary tree sets the position bits from most to least significant.
At level $k$, the integer $u$ encodes the first $k$ bits, with $u=0$ for
the empty prefix at $k=0$. This prefix describes the integers from
$u2^{B-k}$ to $(u+1)2^{B-k}-1$. Its count $c_m(k,u)$ is the number of
these integers in the target range $[m]=\{0,\ldots,m-1\}$.
For $k<B$, the parent and child counts are
\[
 \begin{aligned}
 c_m(k,u)&=\min\{\max\{m-u2^{B-k},0\},2^{B-k}\},\\
 a&=c_m(k+1,2u),\quad b=c_m(k+1,2u+1).
 \end{aligned}
\]
Thus $a$ and $b$ count the target integers whose next bit is $0$ and $1$,
respectively, and $a+b=c_m(k,u)$. At a nonempty node, the next qubit
starts in $\ket0$ and must be mapped to
\[
 \sqrt{\frac{a}{a+b}}\ket0+\sqrt{\frac{b}{a+b}}\ket1.
\]
Choosing these conditional amplitudes at every level gives equal
amplitudes at all $m$ target leaves.

Exact count comparisons handle every node that needs no angle
approximation. When $b=0$, we use the identity $I$, including at empty
nodes where $a=b=0$. When $a=0<b$, we use the bit flip $X$.
This latter case cannot occur for the initial segment $[m]$, but specifies
the gate convention. Nodes with $a=b>0$ use the Hadamard gate $\Had$.
Only unequal positive counts require an approximate rotation
$R_y(2\theta)$, with $\theta=\arctan\sqrt{b/a}$.
Here $R_y(\phi)=e^{-i\phi Y/2}$, where $Y$ is the Pauli $Y$ operator.

The equal split requires only $\Had$, since
$\Had\ket0=R_y(\pi/2)\ket0$. These gates differ on $\ket1$, so we use
$\Had$ consistently in both circuits and $\Had^\dagger=\Had$ in their
inverses. In \cref{fig:uniform-state}(a), the root has child counts
$4$ and $2$ and uses $R_y(2\theta)$. The node with prefix $1$ has child
counts $2$ and $0$, so it uses $I$.

For fixed $Z$, gate selection reduces to locating the boundary of
$[Z]$ at each level. At depth $k$, write
$Z=q_k2^{B-k}+r_k$ with $0\le r_k<2^{B-k}$.
Prefixes $u<q_k$ are full and receive a controlled Hadamard gate.
Only $u=q_k$ can be partial, when $r_k>0$, and its counts select the
gate as above. All remaining prefixes are empty and use $I$.
Since $q_k$ and $r_k$ are classical constants, we can precompute these
gates and select them by comparing the $k$-bit prefix with $q_k$.
The comparisons, exact controls, and their reversals cost $O(k+1)$
gates per level, giving $O(B^2)$ exact gates over all $B$ levels.

Only a partial prefix with two nonempty children requires rotation
synthesis. There are $\operatorname{nz}(Z)-1\le B-1$ such prefixes,
where $\operatorname{nz}(Z)$ counts the nonzero binary digits of $Z$.
Let $\mathsf S_{\rm rot}(\eta)$ denote the gate count for a controlled
rotation with operator-norm error at most $\eta$. Adding these rotations
at the assigned precision to the exact controls gives
\begin{equation}
 C_Z(\delta)=O\!\left(B^2+
 B\mathsf S_{\rm rot}\!\left(\frac{\delta}{cB}\right)\right)
 \label{eq:appendix-UZ-rotation-cost}
\end{equation}
for a sufficiently large absolute constant $c$. When $Z=1$ or $Z$ is a
power of two, no rotation requires synthesis.

\subsection{Coherent angle computation and controlled rotations}
\label{app:controlled-synthesis}

\subsubsection{Computing, applying, and uncomputing the angle}
\label{app:coherent-angle-use}

For variable fiber volume, the tree gates depend on the quantum input
$m=v_y$ and the prefix $u$. Listing a separate tree for every possible
$B$-bit range would introduce exponentially many cases. We instead
compute each angle with the same reversible arithmetic circuit for all
$(m,u)$, allowing both registers to remain in superposition.

At level $k$, exact count comparisons select $I$, $X$, or $\Had$ unless
$a,b>0$ and $a\ne b$. In the remaining case, we compute an approximate
angle $\widetilde\theta$ from $(m,k,u)$ in a register with one integer
bit and $q\ge1$ fractional bits. This register controls rotations on the
next position qubit and is then uncomputed, following Grover and
Rudolph~\cite{groverrudolph}. The three stages shown in
\cref{fig:uniform-state}(b) work as follows.

\begin{enumerate}
\item \emph{Compute the angle coherently.}
A reversible circuit $\mathcal A_k$ evaluates the counts and computes the
binary encoding of $\widetilde\theta$, an approximation to
$\vartheta_{m,k,u}$, where
$\vartheta_{m,k,u}=\arctan\sqrt{b/a}$.
Although $\widetilde\theta$ approximates the desired angle, the circuit
computes its finite binary encoding exactly. It copies this encoding to
the angle register and reverses the intermediate arithmetic, leaving $m$
and $u$ unchanged and returning every arithmetic ancilla qubit to $\ket0$.
\Cref{app:angle-computation,app:fast-angle-computation} give the arithmetic
circuits and their precision bounds.

\item \emph{Apply rotations controlled by the angle bits.}
Express the computed angle as
$\widetilde\theta=\sum_{j=0}^{q}\beta_j2^{-j}$, where $\beta_j\in\{0,1\}$.
For each $j$, the qubit holding $\beta_j$ controls $R_y(2^{1-j})$ on
the next position qubit. These rotations share an axis, so their ideal
product is $R_y(2\widetilde\theta)$. Although $\widetilde\theta$ depends
on the quantum input, the elementary angles $2^{1-j}$ are fixed
constants whose controlled rotations can be synthesized classically.
The construction in \cref{app:controlled-rotation-synthesis} approximates
each rotation within its assigned error while leaving the control bit
unchanged, as required for the next step.

\item \emph{Uncompute the angle.}
Apply $\mathcal A_k^\dagger$ after the controlled rotations. Since the
computational-basis values of $m$, $u$, and the angle bits are unchanged, this returns
the angle register to $\ket{0^{q+1}}$ even when it is entangled with
the target qubit. Every arithmetic ancilla qubit also
returns to $\ket0$.
\end{enumerate}

For an arbitrary state $\ket\chi$ of the next position qubit, the three
steps have the following action with ideal rotations and the computed
angle $\widetilde\theta$:
\begin{equation}
 \begin{aligned}
 &\ket m\ket u\ket{0^{q+1}}\ket\chi\\
 &\quad\longmapsto
 \ket m\ket u\ket{0^{q+1}}R_y(2\widetilde\theta)\ket\chi.
 \end{aligned}
 \label{eq:coherent-angle-three-steps}
\end{equation}

Linearity extends this action to a superposition of $(m,u)$.
To implement $\mathcal U^\dagger$ on the same inputs, we apply the
levels in decreasing order of $k$. At each level, we compute the same
angle, apply the adjoints of the compiled controlled rotations in
reverse order, and uncompute the angle.

\subsubsection{Synthesis with an exact inactive branch}
\label{app:controlled-rotation-synthesis}

The compute--rotate--uncompute construction requires two exact
properties from the synthesized controlled rotations.
They must act as the identity when the control is zero, so a disabled
rotation cannot populate an empty branch. They must also leave the
control bits unchanged, so the arithmetic that computed those bits can
be reversed. A generic approximation in operator norm need not satisfy
either property exactly.

We obtain both properties by pairing a synthesized rotation with its
exact inverse. For a target $R_y(\varphi)$, let $V=R_y(\varphi/2)$ and
synthesize a one-qubit approximation $\widetilde V$. Use the same gate
sequence and its exact inverse in
\begin{equation}
 \widetilde C_{\varphi}
 =(I\otimes\widetilde V)\,\mathrm{CNOT}
  (I\otimes\widetilde V^{\dagger})\,\mathrm{CNOT}.
 \label{eq:controlled-rotation-construction}
\end{equation}

On control $0$, the two synthesized sequences cancel exactly.
On control $1$, the CNOTs instead give
$\widetilde V X\widetilde V^{\dagger}X$, whose ideal value is
$R_y(\varphi)$ because
$XR_y(\alpha)X=R_y(-\alpha)$.
The synthesized sequence and its inverse contribute error at most
$2\eta$ when $\|\widetilde V-V\|\le\eta$, while their global phases
cancel. Only the CNOTs act on the control, leaving it unchanged and
available for exact uncomputation.

Ross--Selinger synthesis~\cite{rossselinger2016} approximates
a one-qubit $z$ rotation to operator-norm error $\eta/2$ using
$O(\log(2/\eta))$ Clifford+$T$ gates. The total gate count follows from
their logarithmic $T$-count because each block of consecutive one-qubit
Clifford gates has a constant-size realization. Conjugation by a fixed
Clifford gate converts the $z$ rotation to a $y$ rotation with only a
constant number of additional gates. Substituting this approximation for
$V$ in \cref{eq:controlled-rotation-construction} gives error at most
$\eta$, using the synthesized sequence, its inverse, and two CNOTs,
with no additional ancilla qubits. Thus
\begin{equation}
 \mathsf S_{\rm rot}(\eta)=O(\log(2/\eta)).
 \label{eq:logarithmic-rotation-synthesis}
\end{equation}

The gates $I$, $X$, and controlled $\Had$ used for zero or equal
counts are implemented exactly over Clifford+$T$.
For example, conjugating CNOT by a one-qubit realization of
$R_y(-\pi/4)$, up to its cancelling global phase, gives controlled $\Had$.

These exact gates and controls ensure that $\widetilde U_Z\ket0$ has
support only on $[Z]$. Exact gates exclude empty children, while
approximate rotations act only between nonempty children.
Ranking therefore receives only valid source fragments and clears the
original source register exactly, even at finite synthesis accuracy.
Uncomputation also resets the arithmetic
registers and the register holding the fiber volume to $\ket0$.

These exact cleanup steps leave approximation error only in the output
state and the fiber-position register. With error at most $\epsilon/2$ in each uniform-state circuit,
the final joint state is within Euclidean distance $\epsilon$ of the
desired output state tensored with all auxiliary registers in $\ket0$.

\subsection{An explicit angle algorithm with a conservative bound}
\label{app:angle-computation}

The remaining cost is reversible angle computation. We give a schoolbook implementation whose
precision bound holds for every pair of unequal positive child counts $a,b$.
The desired angle is $\theta=\arctan\sqrt{b/a}\in(0,\pi/2)$.

To make the arctangent series converge rapidly, we reduce its argument
below $1/2$. Starting with the ratio of the smaller count to the larger,
set $\ell=\min(a,b)$, $h=\max(a,b)$, and
\begin{equation}
 r=\ell/h,\qquad s=\sqrt r,\qquad
 u=\frac{s}{1+\sqrt{1+s^2}}.
 \label{eq:appendix-angle-range-reduction}
\end{equation}
Then $0\le r,s\le1$ and $0\le u\le\sqrt2-1<1/2$.
The half-angle identity gives
\begin{equation}
 \theta_0=\arctan s=2\arctan u,
 \qquad
 \theta=\begin{cases}
 \theta_0,&b\le a,\\
 \pi/2-\theta_0,&b>a.
 \end{cases}
 \label{eq:appendix-angle-recovery}
\end{equation}

The first square root determines the required working precision because
its input can approach zero: an input error of $2^{-p}$ can become
an output error of order $2^{-p/2}$.
The later operations amplify this error by only a constant, since the
second radicand is at least $1$ and the final denominator is at least $2$.
To obtain $q$ fractional output bits, we therefore retain
\begin{equation}
 p=2q+\lceil\log_2(q+1)\rceil+c_0
 \label{eq:appendix-fixed-point-width}
\end{equation}
fractional working bits. The fixed constant $c_0$ supplies guard bits,
and $c_0=32$ suffices for the bound below. Registers holding the scaled
integers use $L=O(B+q)$ bits. We compute $r$, both square roots, and $u$
by binary division and integer square root, truncating each result to
$p$ fractional bits. For example, for a nonnegative fixed-point value
$x$, we compute $\lfloor\sqrt{\lfloor2^p x\rfloor 2^p}\rfloor$ and
interpret the result with $p$ fractional bits.

The reduced argument allows a polynomial with $K=O(q)$ terms to
approximate $\arctan u$. We use
\begin{equation}
 A_K(u)=\sum_{j=0}^{K-1}\frac{(-1)^j u^{2j+1}}{2j+1}.
 \label{eq:appendix-arctan-polynomial}
\end{equation}

With $K=q+4$, Horner's rule in $u^2$ evaluates the polynomial in $O(q)$
multiplications and additions using coefficients $1/(2j+1)$ precomputed
at the working precision. We use the result and a value of $\pi/2$
precomputed to the same precision to recover $\theta$ by
\cref{eq:appendix-angle-recovery}, then round to the nearest value with
$q$ fractional bits.

The initial division contributes absolute error at most $2^{-p}$, and
the inequality
$|\sqrt x-\sqrt y|\le\sqrt{|x-y|}$, valid for $x,y\ge0$,
limits the propagated division error to $2^{-p/2}$ after the first square root, even
when $\ell/h$ is exponentially small. Truncating the square root adds
at most $2^{-p}$ and keeps its value in $[0,1]$. The next radicand lies
in $[1,2]$, where the square root is Lipschitz. The final division has
denominator at least $2$, and the map
$s\mapsto s/(1+\sqrt{1+s^2})$ has bounded derivative on $[0,1]$.
These operations therefore amplify the first square-root error by at
most a constant independent of the counts.

The computed argument also remains within the polynomial's convergence
interval: although $u$ can exceed $\sqrt2-1$ by $O(2^{-p})$, the guard
bits keep it in $[0,1/2]$, where the Taylor remainder satisfies
\begin{equation}
 |\arctan u-A_K(u)|
 \le\frac{u^{2K+1}}{2K+1}\le2^{-(2K+1)}.
 \label{eq:appendix-arctan-error}
\end{equation}

The same interval bounds Horner's multiplier by $u^2\le1/4$, so
coefficient and rounding errors propagate as a geometric sum and are
conservatively bounded by $O(K2^{-p})$.
Combining the errors in range reduction, polynomial evaluation, angle
recovery, and final rounding gives
\begin{equation}
 \begin{split}
 |\widetilde\theta-\theta|\le{}&4\,2^{-p/2}
       +32(K+1)2^{-p}\\
       &+2^{-2K}+2^{-q-1}.
 \end{split}
 \label{eq:explicit-angle-error}
\end{equation}

The four terms account for the first square root, the remaining
fixed-point arithmetic and constants, the Taylor remainder, and the
final rounding. With $K=q+4$, $c_0=32$, and $q\ge1$, their sum is
smaller than $2^{-q}$. Choosing
$q=\lceil\log_2(cB/\delta)\rceil$ therefore keeps the angle error within
its share of the level tolerance for a sufficiently large absolute
constant $c$. The precision depends only on $B$ and $\delta$, so all
count pairs use the same register width.

Schoolbook arithmetic gives an explicit reversible implementation of
this calculation. Division and digit-by-digit integer square root each
use $O(L)$ stages of $O(L)$-bit arithmetic. We retain each stage's
overwritten data, copy the result, and reverse the stages. Each routine
then uses $O(L^2)$ gates and at most $O(L^2)$ ancilla qubits. For Horner
evaluation, we retain $O(q)$ intermediate $L$-bit values and use clean
schoolbook multiplication at each step. Copying the final angle and
reversing the intermediate calculations implements $\mathcal A_k$.
Together, $\mathcal A_k$ and $\mathcal A_k^\dagger$ have total gate count
and maximum arithmetic ancilla count
\begin{equation}
 \begin{aligned}
 C_{\rm angle}(B,q)&=O(qL^2),\\
 W_{\rm angle}&=O(L^2+qL)=O(L^2).
 \end{aligned}
 \label{eq:appendix-angle-cost}
\end{equation}

The $O(q)$ Horner steps dominate the gate count, while the stored division
and square-root histories dominate the ancilla count.

Combining angle computation with rotation synthesis gives the cost of
the variable-range circuit. At each level, the $q+1$ angle bits control
$O(q)$ fixed rotations. We synthesize each to error $O(\delta/(Bq))$
using \cref{app:controlled-rotation-synthesis}, then uncompute the counts
and angle exactly. Summing over $B$ levels gives
\begin{equation}
 C_{\rm var}(\delta)
 =O\!\left(B q\left[L^2+
 \mathsf S_{\rm rot}\!\left(\frac{\delta}{cB q}\right)\right]\right).
 \label{eq:appendix-variable-uniform-cost}
\end{equation}
The exact count operations cost $O(B)$ gates per level and are absorbed
in this bound. Reusing the arithmetic registers at every level keeps
the ancilla count at $O(L^2)$. These bounds include angle computation,
controlled rotations, and arithmetic uncomputation. A concrete $T$-count
or depth requires choosing and compiling specific arithmetic and
synthesis circuits.

\subsection{Implementation using fast arithmetic}
\label{app:fast-angle-computation}

Fast arithmetic reduces the $L^2$ cost of the schoolbook routines while
preserving their computed angles and error bound. To obtain the same
truncated values, we approximate division and square root by Newton
iteration, then correct each result with exact comparisons.
We keep $q=\lceil\log_2(cB/\delta)\rceil$ and choose $L=\Theta(B+q)$ large
enough for all scaled operands, so $\log(2L)=O(q)$.

Normalization places the operands in fixed intervals where Newton
iteration converges quadratically. For division, normalize the nonzero
divisor by a power of two to obtain
$y\in[1,2)$ and approximate $1/y$. For a square root, normalize the
positive radicand by a power of four to obtain $y\in[1,4)$ and
approximate $1/\sqrt y$. The respective Newton iterations are
\begin{equation}
 \begin{aligned}
 x_{j+1}&=x_j(2-yx_j),\\
 z_{j+1}&=\tfrac12 z_j(3-yz_j^2),
 \end{aligned}
 \label{eq:fast-newton-iterations}
\end{equation}
with $x_0=z_0=1/2$. These standard iterations use a constant number of
multiplications and additions per step~\cite[Secs.~3.4--3.5 and
4.2]{brentzimmermann2010}. Their convergence follows from the residuals:
in exact arithmetic, the reciprocal residual
$e_j=1-yx_j$ satisfies $e_{j+1}=e_j^2$ and $0\le e_0\le1/2$.
For the reciprocal square root, $f_j=1-yz_j^2$ satisfies
$f_{j+1}=f_j^2(3+f_j)/4\le f_j^2$ and $0\le f_0\le3/4$.
Thus $O(\log(2L))$ exact iterations give any prescribed constant multiple
of $L$ bits of accuracy.

To retain this convergence at finite precision, we evaluate each update
exactly on the current scaled integers and round down once to $c_1L$
fractional bits. Downward rounding keeps the iterates below the exact
reciprocal or reciprocal square root, where both update maps have
derivative at most $2$. Truncation therefore contributes at most
$2^J2^{-c_1L}$ after $J$ updates, in addition to the exact-iteration
residual of at most $(3/4)^{2^J}$.
Choosing $J=\lceil\log_2(c_2L)\rceil$ and sufficiently large absolute
constants $c_1,c_2$ makes both contributions small enough that the
quotient or square-root estimate has error less than $1/4$, even after
multiplication and rescaling by factors of at most $2^{O(L)}$.

Exact comparisons recover the required integer floors from these
approximations. Since each estimate has error less than $1/4$, its
integer part differs from the required floor by at most one.
For a division $U/V$, a constant number of product comparisons and
corrections gives the integer $t$ satisfying
$tV\le U<(t+1)V$. For a square root of $U$, square comparisons and
corrections give $s$ satisfying $s^2\le U<(s+1)^2$.
These are exactly the truncated results required by
\cref{app:angle-computation}. Zero numerators and radicands receive zero
directly. On branches where a routine is not needed, we supply a nonzero
dummy divisor or a positive dummy radicand. This keeps every arithmetic
operation defined on inputs in superposition.

We implement these division and square-root routines with clean
reversible multiplication. Let
$\Mcost(L)$ and $\mathsf W(L)$ bound its gate count and internal ancilla
count for two $L$-bit inputs. The circuit preserves both inputs, XORs the
full product into a $2L$-qubit output register, and returns its internal
ancilla qubits to $\ket0$. We use the bounds
\[
 \begin{aligned}
 \Mcost(L)&=O\!\left(L\log(2L)\right),\\
 \mathsf W(L)&=O\!\left(L\log(2L)\right),
 \end{aligned}
\]
obtained from fast integer multiplication and reversible
compute--copy--uncompute~\cite{harvey2021integer,ragavan2023space}.

To implement the full routine coherently, we first find the leading
nonzero bit and normalize with controlled shifts, using
$O(L\log(2L))$ gates. We then perform the fixed number of Newton
iterations and the final corrections, using clean multiplication and
retaining the intermediate values. Every register width and iteration
count depends only on $L$, and comparisons control reversible operations
without measurement, so one circuit handles all inputs in superposition.
Copying its result and reversing the calculation gives, for each
division or square root,
\begin{equation}
 \begin{aligned}
 C_{\rm div/sqrt}(L)&=O\!\left(\Mcost(L)\log(2L)\right),\\
 A_{\rm div/sqrt}(L)&=O\!\left(\mathsf W(L)+L\log(2L)\right),
 \end{aligned}
 \label{eq:fast-division-sqrt-cost}
\end{equation}
where $A_{\rm div/sqrt}$ counts ancilla qubits. Multiplication reuses
$\mathsf W(L)$ ancilla qubits between iterations, while the retained
values occupy $O(L\log(2L))$ qubits. Reversing the calculation returns
all these ancilla qubits to $\ket0$.

Each angle requires a constant number of these divisions and square
roots, followed by the $O(q)$ Horner steps in
\cref{eq:appendix-arctan-polynomial}. Retaining the $O(q)$ intermediate
values and reusing the ancilla qubits for multiplication gives
\begin{equation}
 \begin{aligned}
 C_{\rm angle}^{\rm fast}(B,q)
   &=O\!\left((q+\log(2L))\Mcost(L)\right)\\
   &=O\!\left(qL\log(2L)\right),\\
 A_{\rm angle}^{\rm fast}(B,q)
   &=O\!\left(qL+L\log(2L)+\mathsf W(L)\right)\\
   &=O(qL).
 \end{aligned}
 \label{eq:fast-angle-cost}
\end{equation}
The second equalities use $\log(2L)=O(q)$ and the multiplication bounds
stated above. Copying the angle, reversing
its computation, and uncomputing it after the rotations change the gate
count by only a constant factor. Since the truncated angle is unchanged,
the same rounding and synthesis error bounds apply. Summing over the
$B$ levels gives
\begin{equation}
 C_{\rm var}^{\rm fast}(\delta)
 =O\!\left(Bq\left[L\log(2L)+
 \mathsf S_{\rm rot}\!\left(\frac{\delta}{cBq}\right)\right]\right),
 \label{eq:fast-variable-uniform-cost}
\end{equation}
with $O(qL)$ ancilla qubits reused across all levels.
Fast arithmetic therefore replaces the $L^2$ term in
\cref{eq:appendix-variable-uniform-cost} by $L\log(2L)$ while retaining
the same error guarantee and exact arithmetic cleanup.
The resulting gate count is polynomial in $B$ and $\log(1/\delta)$,
using one circuit for all ranges in superposition.

\section{Input accuracy and a sufficient arithmetic width}
\label{app:input-accuracy}

\subsection{From real coefficients to an integer polynomial}

The construction in Section~\ref{sec:polynomial-input} takes an integer
polynomial $P$ as input. This subsection shows how to choose its coefficients
when the target state has amplitudes proportional to samples of a polynomial
$f$ with real coefficients. We bound the change in the sampled amplitudes
and then the effect of normalization, obtaining a sufficient coefficient
bit length $\beta$ for a prescribed state error $\epsilon_{\mathrm{coef}}$.
This supplies the coefficient precision used in \Cref{cor:precision-regimes}.
The sum of this coefficient error and the circuit approximation error in
\Cref{thm:fast-clean-preparation} bounds the total preparation error.

Let $f(x)=\sum_{k=0}^d c_kx^k$ have a nonzero sampled vector, set
$C=\max_k|c_k|$, and choose the nearest integers $a_k$ to $2^bc_k/C$.
Then $|a_k/2^b-c_k/C|\le2^{-b-1}$, with $p$ and $P$ as in
\cref{eq:dyadic-p,eq:integer-P}.

To make the precision dependence explicit, let
$\mathbf f:=(f(i/N))_{i\in[N]}\ne0$ and define the scale-invariant sampling
condition
\begin{equation}
 \chi_N(f):=\frac{\sqrt N\,C}{\|\mathbf f\|_2}.
\label{eq:sample-conditioning}
\end{equation}
For $\mathbf p:=(p(i/N))_{i\in[N]}$, coefficientwise rounding gives
\begin{equation}
 \left\|\frac{\mathbf f}{C}-\mathbf p\right\|_2
 \le \sqrt N(d+1)2^{-b-1}.
\label{eq:dyadic-grid-error}
\end{equation}
For nonzero vectors $x$ and $y$, normalization obeys
\begin{equation}
 \left\|\frac{x}{\|x\|_2}-\frac{y}{\|y\|_2}\right\|_2
 \le\frac{2\|x-y\|_2}{\|x\|_2}.
\label{eq:normalization-stability}
\end{equation}
Since the normalized grid state of $\mathbf p$ is the state $\ket{\psi}$ in
\cref{eq:psi-P}, \cref{eq:dyadic-grid-error} implies
\begin{equation}
 \left\|
 \frac{1}{\|\mathbf f\|_2}\sum_{i\in[N]}f(i/N)\ket{i}
 -\ket{\psi}
 \right\|_2
 \le(d+1)\chi_N(f)2^{-b}.
\label{eq:dyadic-state-error}
\end{equation}
Thus, for $0<\epsilon_{\mathrm{coef}}<1$, it suffices to choose
\begin{equation}
 b=\left\lceil\log_2\frac{(d+1)\chi_N(f)}{\epsilon_{\mathrm{coef}}}\right\rceil,
 \qquad
 \beta\le b+1,
\label{eq:coefficient-approximation-width}
\end{equation}
which also ensures $\mathbf p\ne0$: the unnormalized perturbation is
at most $\epsilon_{\mathrm{coef}}/2$ times $\|\mathbf f/C\|_2$.
The coefficient bit length therefore satisfies
\begin{equation}
 \beta=O\!\left(\log\frac{(d+1)\chi_N(f)}{\epsilon_{\mathrm{coef}}}\right).
\label{eq:beta-epsilon-scaling}
\end{equation}
If the sampled profile is well conditioned in the sense that
$\chi_N(f)=\operatorname{poly}(d)$, then
$\beta=O(\log((d+1)/\epsilon_{\mathrm{coef}}))$.  Hence the dependence of $\beta$ on the
target precision is logarithmic; any further dependence on $N=2^n$ enters
through $\chi_N(f)$.

This coefficient estimate assumes access to the required approximations of
the real coefficients and to a sufficient upper bound on $\chi_N(f)$.
It bounds how many coefficient bits suffice, but does not specify how to
obtain those bits from arbitrary real coefficients.

The integer polynomial is the input to the exact combinatorial construction.
Once its coefficients are fixed, changing the accuracy of the later circuit synthesis
does not change its alias table. Approximating a new real polynomial at a
new coefficient accuracy may, of course, change those coefficients and all
preprocessed data.

\subsection{Bounding all exact intermediate values}
\label{app:arithmetic-width-proof}

Once the integer coefficients are fixed, the implicit alias oracle and
fiber ranking circuits evaluate $P,Q,S,D,E$ using exact integer arithmetic.
This subsection proves the sufficient magnitude width
$B=O(\beta+dn+d\log(d+1))$ stated in \cref{eq:arithmetic-width}
and gives an explicit choice of $B$. We bound the polynomial values and
finite differences, followed by the products and signed partial sums in
the cumulative functions. The bound also covers the capacity $Z$ and the
fiber volumes $v_y$ used in state preparation. These estimates justify
the register sizes and arithmetic gate costs used in
\Cref{thm:implicit-representation,thm:fast-clean-preparation}.

By \cref{eq:coefficient-width}, $|a_k|<2^\beta$ for every
$k\in\{0,\ldots,d\}$. Therefore, for every
$i\in[N]$,
\begin{equation}
 |P(i)|<(d+1)2^\beta N^d.
\label{eq:P-bound}
\end{equation}
The coefficients $\kappa_k$ require only the values $Q(u)$ for
$0\le u\le2d$. For these integers,
\begin{equation}
 |P(u)|<(d+1)2^\beta\bigl[N(2d+1)\bigr]^d,
\label{eq:P-difference-bound}
\end{equation}
and \cref{eq:finite-difference-coefficients} gives
\begin{equation}
 |\kappa_k|\le 2^k\max_{0\le u\le k}Q(u),
 \qquad 0\le k\le2d.
\label{eq:kappa-bound}
\end{equation}
For $0\le t\le N$ and $0\le k\le2d$, we also have
$\binom{t}{k+1}\le N^{k+1}$. These bounds show that there is a sufficient
width $B$ for exact arithmetic with
\begin{equation}
 B=O\bigl(\beta+dn+d\log(d+1)\bigr)
\label{eq:appendix-arithmetic-width}
\end{equation}
such that the absolute value of every signed integer arising in
\cref{eq:S-binomial,eq:D-closed,eq:E-closed}, including intermediate products,
is smaller than $2^{B}$. A signed intermediate uses $B$ qubits to store
its absolute value and one additional qubit to store its sign. Indeed,
$Z\le N\max_{i\in[N]}Q(i)$ and
$0\le D(r),E(s)\le NZ$, so the cumulative values also have this width. We fix
such a width $B$, increasing it if necessary so that
$B\ge n+1$; thus an unsigned $B$-qubit register can also
represent the endpoint $N$. We further increase $B$, if needed, so
that $NZ<2^{B}$. It follows that $Z<2^{B}$ and
$v_y<2^{B}$ for every $y\in[N]$.

For example, the following conservative magnitude width is sufficient for
the polynomial and cumulative calculations above:
\begin{equation}
 \begin{split}
 B={}&2\beta+(4d+2)n+2d+8\\
 &+(2d+1)\lceil\log_2(2d+1)\rceil
   +2\lceil\log_2(d+1)\rceil.
 \end{split}
 \label{eq:explicit-arithmetic-width}
\end{equation}
This is not a minimal register allocation. It follows by bounding each
$\kappa_k\binom{t}{k+1}$ separately, allowing their signed partial sums,
and including the extra factor $N$ in $D$ and $E$.
An extra endpoint bit is allowed when a later circuit for uniform states
represents $2^{B}$.

\begin{remark}[Scale of exact integer samples]
\label[remark]{rem:exact-integer-scale}
Let $g$ have exact degree $r\ge1$ and leading coefficient $g_r\ne0$, with
$N=2^n\ge r+1$. If an integer $b_{\rm int}\ge0$ makes
$2^{b_{\rm int}}g(i/N)$ integral for every $i\in[N]$, then
\[
 b_{\rm int}\ge rn-\log_2(r!\,|g_r|).
\]
Indeed, the $r$th forward difference of these integer samples at zero is
$r!\,2^{b_{\rm int}-rn}g_r$, a nonzero integer whose absolute value is
at least one. Here $b_{\rm int}$ is the exponent of the scale that makes
the samples integral. For $g=p^2$, with $p$ of exact degree $d$ and
$N\ge2d+1$, the bound is
$2dn-\log_2((2d)!\,|g_{2d}|)$, or $2dn-O(d\log(d+1))$ when
$|g_{2d}|=\Theta(1)$. The weights in this construction use
$Q(i)=2^{2b+2dn}p(i/N)^2$.
\end{remark}

\bibliography{references}

\end{document}